%% file: sup-attack-yao.tex
\documentclass{llncs}

\usepackage{makeidx}  
\usepackage{color}
\usepackage{graphicx}
\usepackage{hyperref}
\usepackage{physics}
\usepackage{mathtools}
\usepackage{float}
\usepackage{amsmath}
\usepackage{amsfonts}
\usepackage{enumitem}
\usepackage{algcompatible}
\usepackage{algorithm}
\usepackage{xfrac}
\usepackage{qcircuit}
\floatname{algorithm}{Protocol}

\newcommand{\BQP}{\mathsf{BQP}}
\newcommand{\BPP}{\mathsf{BPP}}
\newcommand{\QPPT}{\mathsf{QPPT}}
\newcommand{\Enc}{\mathsf{Enc}}
\newcommand{\Dec}{\mathsf{Dec}}

\newcommand{\Abort}{\mathsf{Abort}}
\newcommand{\CNOT}{\mathsf{CNOT}}
\newcommand{\X}{\mathsf{X}}
\newcommand{\Z}{\mathsf{Z}}
\newcommand{\Ha}{\mathsf{H}}
\newcommand{\Id}{\mathsf{1}}
\newcommand{\GHZ}{\mathsf{GHZ}}
\newcommand{\Pone}{P_1}
\newcommand{\Ptwo}{P_2}
\newcommand{\Regone}{\mathcal{X}}
\newcommand{\Regtwo}{\mathcal{Y}}

\newcommand{\PredE}{\mathsf{P}_f^{E}}

\floatstyle{ruled}
\newfloat{idealf}{htb!}{idf}
\floatname{idealf}{Ideal Functionality}

\newfloat{attack}{htb!}{idf}
\floatname{attack}{Attack}

\usepackage{todonotes}

\input{add.sty}

\begin{document}

\frontmatter

\pagestyle{plain}

\mainmatter

\title{Dispelling Myths on Superposition Attacks:\\
Formal Security Model and Attack Analyses 
}



\author{Luka Music$^1$, C\'{e}line Chevalier$^{2}$, Elham Kashefi$^{1,3}$}

\institute{$^1$ D\'epartement Informatique et R\'eseaux, CNRS, Sorbonne 
Universit\'e\\
$^2$ Universit\'e Panth\'eon-Assas Paris 2 \\
$^3$ School of Informatics, University of Edinburgh}

\maketitle



\begin{abstract}

\input{tex_files/00_abstract.tex}

\end{abstract}

\section{Introduction}
\label{sec:introduction}
\input{tex_files/01_introduction.tex}

\input{tex_files/01bis_openq.tex}

\section{Preliminaries}
\label{prelims}
\label{sec:preliminaries}
\input{tex_files/02_preliminaries.tex}

\section{New Security Model for Superposition Attacks}
\label{sec:security:model}
\input{tex_files/03_security_model.tex}

\section{The Modified Honest-but-Curious Yao Protocol}
\label{hbcyao}
\label{sec:yao:protocol}
\input{tex_files/04a_yao_protocol.tex}

\section{Analysis of Yao's Protocol with Superposition Access}
\label{supponyao}
\label{sec:attack}
\input{tex_files/05_attack_yao.tex}


\section{Conclusion}
\label{sec:concl}
\input{tex_files/07_conclusion.tex}

\subsection*{Acknowledgments}
\label{subsec:ack}
\input{tex_files/08_ack.tex}


\bibliographystyle{alpha}
\bibliography{biblio,abbrev3,crypto}

\appendix

\section{Additional Quantum Notations}
\label{app:more_q}
\input{tex_files/A0_more_q.tex}

\section{Security Model: Auxiliary Definitions and Results}
\label{app:OTP}
\input{tex_files/A1_OTP.tex}

\section{Formal Modified Yao Protocol}

\subsection{Formal Definitions for Symmetric Encryption Schemes}
\label{subsec:symmetric:encryption}
\input{tex_files/04b_symmetric_encryption.tex}


\subsection{Description, Correctness and $\QPPT$-Security of the Modified Yao Protocol}
\label{app:multiple_output}
%
\label{proofs_CS}
\label{app:proof_sec_yao}
\input{tex_files/04c_new_yao_protocol.tex}

\section{Superposition Attack on the Modified Yao Protocol}

\subsection{Attack Description}
\label{app:form_att}
\input{tex_files/A5_form_att.tex}

\subsection{Proof of Insecurity of the Modified Yao Protocol against Superposition 
Attacks}
\label{app:proof_insec}
\input{tex_files/A6_insec_proofs.tex}

\subsection{Formal Superposition-Secure Yao Protocol}
\label{app:sup_sec_yao}
\input{tex_files/A7_sup_sec_yao.tex}

\section{Attack Optimisation and Application to Oblivious Transfer}
\label{sec:OT}
\input{tex_files/06_attack_OT.tex}

\end{document}

%% file: tex_files/00_abstract.tex

With the emergence of quantum communication, it is of folkloric belief that 
the security of classical cryptographic protocols is automatically broken if 
the Adversary is allowed to perform superposition queries and the honest 
players forced to perform actions coherently on quantum states. Another 
widely held intuition is that enforcing measurements on the exchanged messages is enough 
to protect protocols from these attacks.

    However, the reality is much more complex. 
	Security models dealing with superposition attacks only consider 
unconditional security. Conversely, security models considering 
computational security assume that all supposedly classical messages are 
measured, which forbids by construction the analysis of superposition attacks.
    To fill in the gap between those models, Boneh and Zhandry have started to 
study the quantum computational security for classical primitives in their 
seminal work at Crypto'13, but only in the single-party setting. To the best 
of our knowledge, an equivalent model in the multiparty setting is still 
missing.

	In this work, we propose the first computational security model 
considering superposition attacks for multiparty protocols. We show that our 
new security model is satisfiable by proving the security of the well-known 
One-Time-Pad protocol and give an attack on a variant of the equally reputable 
Yao Protocol for Secure Two-Party Computations. The post-mortem of this attack 
reveals the precise points of failure, yielding highly counter-intuitive 
results: Adding extra classical communication, which is harmless for classical 
security, can make the protocol become subject to superposition attacks. 
We use this newly imparted knowledge to construct the first concrete protocol for 
Secure Two-Party Computation that is resistant to superposition attacks. 
Our results show that there is no straightforward answer to provide for either the 
vulnerabilities of classical protocols to superposition attacks or the adapted 
countermeasures.

%% file: tex_files/01_introduction.tex
Recent advances in quantum technologies threaten the security of many 
widely-deployed cryptographic primitives if we assume that the Adversary has 
classical access to the primitive but can locally perform quantum 
computations. This scenario has led to the emergence of \emph{post-quantum 
cryptography}. But the situation is even worse in the \emph{fully quantum} 
scenario, if we assume the Adversary further has quantum access to the 
primitive and can query the oracle with quantum states in superposition. 
Such access can arise in the case where the Adversary has direct access to the primitive 
that is being implemented (eg.\ symmetric encryption, hash functions), or if a protocol is 
used as a sub-routine where the Adversary plays all roles (as in the Fiat-Shamir transform based on Sigma Protocols) 
and can therefore implement them all quantumly. In the future, various primitives might natively be 
implemented on quantum machines and networks, either to benefit from speed-ups or because 
the rest of the protocol is inherently quantum.
In this case, more information could be leaked, leading to new non-trivial 
attacks, as presented in a series of work initiated in 
\cite{DFNS11,BZ13,kaplan2016breaking}. A possible countermeasure against such 
\emph{superposition attacks} is to forbid any kind of quantum access to the 
oracle through measurements. However, the security would then rely on the 
physical implementation of the measurement tool, which itself could be 
potentially exploited by a quantum Adversary. Thus, providing security 
guarantees in the fully quantum model is crucial.
	We focus here on the multiparty (interactive) setting.

\headingb{Analysis of Existing Security Models.}
Modelling the security of classical protocols in a quantum world (especially 
multiparty protocols) is tricky, since various arbitrages need to be made concerning 
the (quantum or classical) access to channels and primitives.


A first possibility is to consider classical protocols embedded as quantum 
protocols, thus allowing the existence of superposition attacks. However, in 
such a setting, previous results only consider \emph{perfect security}, meaning that the 
messages received by each player do not contain more information than its 
input and output. The seminal papers starting this line of work are those 
proving the impossibility of bit commitment~\cite{Lo_1997,commit2}.
The perfect security of the protocol implies that no additional information is 
stored in the auxiliary quantum registers of both parties at the end of the 
protocol and can therefore be traced out, so that an Adversary can easily 
produce a superposition of inputs and outputs.

This is for example the approach of \cite{DFNS11}, and \cite{SSS15}, where the 
perfect correctness requirement is in fact a perfect (unconditional) security 
requirement (the protocol implements the functionality and \emph{only} the 
functionality). In \cite{DFNS11}, they consider an even more powerful 
adversarial scenario where not only the honest player's actions are described 
as unitaries (their inputs are also in superposition) but the Adversary can 
corrupt parties in superposition (the corruption is modelled as an oracle call 
whose input is a subset of parties and which outputs the view of the 
corresponding parties). Both papers show that protocols are insecure in such a 
setting: In \cite{DFNS11}, they show that in the case of a multi-party 
protocol implementing a general functionality (capable of computing any 
function), no Simulator can perfectly replicate the superposition of views of 
the parties returned by the corruption oracle by using only an oracle call to 
an Ideal Functionality. In the case of a deterministic functionality, they 
give a necessary and sufficient condition for such a Simulator to exist, but 
which cannot be efficiently verified and is not constructive. In \cite{SSS15}, 
they prove that any non-trivial Ideal Functionalities that accept 
superposition queries (or, equivalently, perfectly-secure protocols emulating 
them) must leak information to the Adversary beyond what the classical 
functionality does (meaning that the Adversary can do better than simply 
measure in the computational basis the state that it receives from the 
superposition oracle). In both cases, they heavily rely on the assumption of 
unconditional security to prove strong impossibility results and their proof 
techniques cannot be applied to the computational setting. 



The second possibility to model the security of classical protocols in a 
quantum world is to define purely classical security models, in the sense that 
all supposedly classical messages are measured (Stand-Alone Model of 
\cite{HSS15} or the Quantum UC Model of \cite{Unr10}).
Some (computationally) secure protocols exist in this setting, as shown by a 
series of articles in the literature (eg.\ \cite{LKHB17}). However, these 
models forbid by construction the analysis of superposition attacks, precisely since all 
classical communications are modelled as measurements. 

\paragraph{The missing link.} The results of \cite{SSS15,DFNS11} in
the unconditional security setting are not directly applicable to a 
\emph{Computationally-Bounded Adversary}. The premiss to their analyses is 
that since the perfect execution of non-trivial functionalities is insecure,
any real protocol implementing these functionalities is also insecure against 
Adversaries with quantum access 
(even more since they are simply 
computationally secure). However it turns out that, precisely because the 
protocol is only computationally-secure, the working registers of the parties 
cannot be devoid of information as is the case in the perfectly-secure setting 
(the messages contain exactly the same information as the secret inputs of the 
parties, but it is hidden to computationally-bounded Adversaries) and the 
techniques used for proving the insecurity of protocols in the perfect 
scenario no longer work.

This issue has been partially solved for single-party protocols with oracle 
queries in the line of work from~\cite{BZ13}, but never extended fully to the 
multi-party setting. The difficulty arises by the interactive property of such 
protocols. Indeed, in a real protocol, more care needs to be taken in 
considering all the registers that both parties deal with during the execution 
(auxiliary qubits that can be entangled due to the interactive nature of the 
protocols). Furthermore, care must also be taken in how the various classical 
operations are modelled quantumly, as choosing standard or minimal oracle 
representations may influence the applicability of some attacks \cite{KKVB02}. 
The naive implementation of superposition attacks, applied to a real-world 
protocol, often leads to a joint state of the form 
$\sum\limits_{x, m_1, m_2}\ket{x}\ket{m_1}\ket{m_2}\ket{f(x, y)}$
for a given value $y$ of the honest player's input, and with the second 
register (containing the set of messages $m_1$ sent by the Adversary) being in 
the hands of the honest player ($m_2$ is the set of messages sent by the 
honest player and $f(x, y)$ is the result for input $x$). This global state 
does not allow the known attacks
(such as \cite{kaplan2016breaking}) to go through as the message registers 
cannot simply be discarded. This 
shows that the simple analysis of basic ideal primitives in the superposition 
attack setting is not sufficient to conclude on the security of the overall 
computationally-secure protocol and motivates the search for a framework for 
proving security of protocols against such attacks.

\headingb{Our Contributions.}
The main purpose of this paper is thus to bridge a gap between two settings: 
one considers the security analysis of superposition attacks, but 
either for perfect security~\cite{DFNS11,SSS15} (both works preclude the existence of 
secure protocols by being too restrictive)
or 
only for single-party primitives with oracle 
access~\cite{BZ13}, while the other explicitly forbids 
such attacks by 
measuring classical messages~\cite{Unr10,HSS15}

To our knowledge, our result is the first attempt to formalise a security 
notion capturing security of two-party protocols against superposition attacks 
with computationally-bounded Adversaries as a simulation-based definition. We 
consider a more realistic scenario where a computational Adversary 
corrupts a fixed set of players at the beginning of 
the protocol and the input of the honest players are fixed classical values. 
We suppose that the ideal world trusted third party always measures its 
queries (it acts similarly to a classical participant), while 
the honest player always performs actions in superposition unless specifically 
instructed by the quantum embedding of the protocol (the Adversary and the 
Simulator can do whatever they want). Security is then defined by considering 
that an attack is successful if an Adversary is able to distinguish between 
the real and ideal executions with non-vanishing probability. The reason for adding 
a measurement to the functionality is to enforce that the (supposedly classical) 
protocol behaves indeed as a classical functionality. This is further motivated by the results 
of previous papers proving that functionalities with quantum behaviour are inherently broken.


\paragraph{Case Studies.} We show that our proposed security model is 
satisfiable by proving the superposition-resistance of the classical One-Time-Pad protocol for 
implementing a Confidential Channel. Conversely, we also present an attack on a slight 
variant of the Honest-but-Curious\footnote{An Adversary is Honest-but-Curious if it 
acts honestly during the protocol but performs arbitrary computations later to 
recover more information about the input of the honest player.} version of the classical Yao's protocol 
\cite{Yao86} for Secure Two-Party Computation. On the other hand, it is 
secure against $\QPPT$ Adversaries (that have a quantum computer internally 
but send classical messages), therefore showing a separation. The variant is presented to 
demonstrate unusual and counter-intuitive reasons for which protocols may be 
insecure against superposition attacks. 

\paragraph{Proof Technique.} 
During the superposition attack, the Adversary 
essentially makes the honest player implement the oracle call in Deutsch-Jozsa's 
(DJ)
algorithm \cite{deutsch1992rapid} through its actions on a superposition 
provided by the Adversary. The binary function for which this oracle query is 
performed is linked to two possible outputs of the protocol.
The Adversary can then apply the rest of the DJ algorithm to decide the nature of 
the function\footnote{The DJ algorithm decides whether a binary 
function is balanced or constant}, which allows it to extract the XOR of the two outputs. 
Similarly to the DJ algorithm where the state containing the output 
of the oracle remains in the $\ket{-}$ state during the rest of the algorithm 
(it is not acted upon by the gates applied after the oracle call), the 
Adversary's actions during the rest of the attack do not affect the output 
register. Interestingly, this means that the attack can thus also be performed on the same protocol but where 
the Adversary has no output.

\paragraph{Superposition-Secure Two-Party Computation.} 
Counter-intuitively, it is therefore not the output that makes the attack 
possible, but in this case the attack vector is a message consisting of information that, classically, the 
Adversary should already have, along with a partial measurement on the part of 
the honest player (which is even stranger considering that it is usually 
thought that the easiest way to prevent superposition attack is to measure the 
state). 
This shows that adding extra communication, even an exchange of classical information which seems 
meaningless for classical security, can make the protocol become subject to 
superposition attacks.
Removing the point of failure by never sending back this information 
to the Adversary (as is the case in the original Yao Protocol) makes the 
protocol very similar in structure to the One-Time-Pad Protocol, where one 
party sends everything to the other, who then simply applies local operations. 
The proof for the One-Time-Pad works by showing that there is a violation of 
the no-signalling condition of quantum mechanics if the Adversary is able to 
distinguish between ideal and real scenarios (if it were able to gain any information, 
it would be solely from these local operations by the honest player, which would imply that information 
has been transferred faster than the speed of light).
This technique can only be reused if the honest party in Yao's protocol does not give away the 
result of the measurement on its state (by hiding the fact that it either 
succeeded in completing the protocol or aborted if it is unable to do so 
correctly). We show that Yao's protocol is secure against superposition attacks 
if the (honest) Evaluator recovers the output and does not divulge whether or not it has aborted.



\headingb{Contribution Summary and Outline.} After basic notations in Section~\ref{sec:preliminaries}:
\begin{itemize}
\item Section~\ref{sec:security:model} gives a new security model for superposition attacks;
\item Section~\ref{sec:yao:protocol} proves the security of a variant of Yao's protocol against adversaries exchanging classical messages;
\item Section~\ref{state_gen} demonstrates a superposition attack against this same protocol, applied in Appendix~\ref{sec:OT} to an Oblivious Transfer protocol with slightly improved attack success probability;
\item Section~\ref{subsec:path} builds a superposition-resistant version of 
Yao's protocol by leveraging the knowledge acquired through the attack.
\end{itemize}


%% file: tex_files/01bis_openq.tex
\headingb{Open Questions.}
An interesting research direction would be to analyse what functionalities (if any) can be implemented using the ``insecure" ideal functionalities with allowed superposition access described in \cite{SSS15}. Since these functionalities necessarily leak information, they can no longer be universal: if they were, then it would be possible to construct non-leaky functionalities with protocols only making calls to these leaky functionalities. However, some limited functionalities may also be useful, as exemplified by the biased coin-toss.

The security model presented in this paper does not support any kind of composability, as can be shown with rather simple counter-examples. While it would be ideal to have a simulation-based fully-composable framework for security against superposition attacks, we leave this question open for now. 

While we prove that Yao's protocol is secure in our model if the Evaluator does not reveal the outcome of the protocol, it would also be interesting to analyse the consequence of removing the minimal oracle assumption from the symmetric encryption scheme and instead use a traditional IND-CPA symmetric encryption with the original Yao garbled table construction (therefore adding an additional entangled quantum register). The Yao protocol has recently been studied in \cite{BDK20} and found secure against Adversaries that do not have superposition access to the honest party, under the assumption that the encryption scheme is pq-IND-CPA (the quantum Adversary does not make queries to the encryption oracle in superposition but has access to a Quantum Random Oracle).

Finally, this paper show that partial measurements by honest players are not sufficient to prevent superposition attacks. It would be interesting to find the minimum requirements for the security of protocols with superposition access and measurements by honest parties so that they are as secure as classical protocols. This field of study has been somewhat initiated by the work of \cite{Unr15} with the collapsing property (measuring one message makes the other message collapse to a classical value if it passes some form of verification), but the question of whether there is a minimal amount of information that should be measured to be superposition-secure remains open.

%% file: tex_files/02_preliminaries.tex
All protocols will be two-party protocols (between parties $\Pone$ and 
$\Ptwo$). $\Pone$ will be considered as the Adversary (written $\Pone^*$ when corrupted), while $\Ptwo$ is 
honest. Although we consider purely classical protocol, in order to be able to 
execute superposition attacks, both parties will have access to multiple 
quantum registers, respectively denoted collectively $\Regone$ and $\Regtwo$.

All communications are considered as quantum unless specified and we call 
quantum operations any completely positive and trace non-decreasing 
superoperator acting on quantum registers (see \cite{nielsenchuang} and 
Appendix~\ref{app:more_q} for more details), with $\Id_\mathcal{A}$ being 
identity operator on register $\mathcal{A}$.

The principle of superposition attacks is to consider that a player, otherwise honestly behaving, performs all of its operations on quantum states rather than on classical states. In fact, any classical operation defined as a binary circuit with bit-strings as inputs can be transformed into a unitary operation that has the same effect on each bit-string (now considered a basis state in the computational basis) as the original operation by using Toffoli gates. Although any quantum computation can be turned into a unitary operation (using a large enough ancillary quantum register to purify it), it may be that the honest player may have to take a decision based on the value of its internal computations. This is more naturally defined as a measurement, and therefore such operations will be allowed but only when required by the protocol (in particular, when the protocol branches out depending on the result of some computation being correct). The rest of the protocol (in the honest case) will be modelled as unitary operations on the quantum registers of the players (see Appendix \ref{Qembed} for the precise description of the quantum embedding of a classical protocol).

There are two ways to represent a classical function $f : \{0,1\}^n \leftarrow \{0, 1\}^m$ as a unitary operation. The most general way (called \emph{standard oracle} of $f$) is defined on basis state $\ket{x}\ket{y}$ (where $x \in \{0,1\}^n$ and $y \in \{0, 1\}^m$) by $U_f \ket{x}\ket{y} = \ket{x}\ket{y \oplus f(x)}$, where $\oplus$ corresponds to the bit-wise XOR operation. On the other hand, if $n = m$ and $f$ is a permutation over $\{0,1\}^n$, then it is possible (although in general inefficient) to represent $f$ as a \emph{minimal oracle} by $M_f\ket{x} = \ket{f(x)}$. Note that this representation is in general more powerful than the standard representation of classical functions as quantum unitaries (see \cite{KKVB02} for more information).

The security parameter will be noted $\eta$ throughout the paper (it is passed implicitly as $1^{\eta}$ to all participants in the protocol and we omit when unambiguous). A function $\mu$ is \emph{negligible in $\eta$} if, for every polynomial $p$, for $\eta$ sufficiently large it holds that $\mu(\eta) < \frac{1}{p(\eta)}$. For any positive integer $N \in \mathbb{N}$, let $[N] := \{1, \ldots, N\}$. For any element $X$, $\#X$ corresponds to the number of parts in $X$ (eg.\ size of a string, number of qubits in a register). The special symbols $\Abort$ will be used to indicate that a party in a protocol has aborted. 

%% file: tex_files/03_security_model.tex
\headingb{General Protocol Model.} 
We assume that the input of the honest player is classical, meaning it is 
a pure state in the computational basis, unentangled from the rest of the 
input state (which corresponds to the Adversary's input). This is in stark 
contrast with other papers considering superposition attacks 
\cite{SSS15,DFNS11} where the input of the honest players is always 
a uniform superposition over all possible inputs. We also consider that 
the corrupted party is chosen and fixed from the beginning of the protocol. 
We will often abuse notation and consider the corrupted 
party and the Adversary as one entity.


The security of protocols will be defined using the \emph{real/ideal 
simulation paradigm}, adapted from the Stand-Alone Model of \cite{HSS15}. The 
parties involved are: an Environment $\mathcal{Z}$, the parties 
participating in the protocol, a Real-World Adversary $\mathcal{A}$ and an 
Ideal-World Adversary also called Simulator $\mathcal{S}$ that runs~$\mathcal{A}$ 
internally and interacts with an Ideal 
Functionality (that the protocol strives to emulates). 
An execution of the protocol (in the real or ideal case) works as follows:

\begin{enumerate}

\item The Environment $\mathcal{Z}$ produces the input $y$ of $\Ptwo$, 
the auxiliary input state~$\rho_{\mathcal{A}}$ of the Adversary (containing 
an input for corrupted party $\Pone^*$, possibly in superposition).

\item The Adversary interacts with either the honest player performing the 
protocol or a Simulator with single-query access to an Ideal Functionality.

\item Based on its internal state, it outputs a bit corresponding to its guess 
about whether the execution was real or ideal. If secure, no Adversary should 
be able to distinguish with high probability the two scenarios.


\end{enumerate}


\headingb{Adversarial Model.} 
To capture both the security against Adversaries with and without 
superposition (so that we may compare both securities for a given protocol), 
we parametrise the security Definition \ref{sec-def} below with a class of 
Adversaries $\mathfrak{X}$. This class $\mathfrak{X}$ can be either $\BQP$ or 
$\QPPT$ and the Simulator is of the same class as the Adversary.
	A $\BQP$ machine is also called a polynomial quantum Turing machine and 
recognises languages in the $\BQP$ class of complexity \cite{Yao93,nielsenchuang}. 
They can perform any polynomial-sized family of quantum 
circuits and interact quantumly with other participants (by sending quantum 
states which may or may not be in superposition).
	A $\QPPT$ machine on the other hand is a classical machine which can 
perform the same computations as a quantum computer (and therefore it is not 
required to terminate in classical polynomial-time). More formally, no party 
interacting classically with a machine should be able to distinguish whether 
it is a $\BQP$ or a $\QPPT$ machine. The formal definition of complexity class 
$\BQP$ (and by extension of efficient quantum machines) is given in 
Definition~\ref{bqp_m} \cite{Yao93}, while that of a $\QPPT$ 
machine \cite{Unr10} is given in Definition~\ref{qppt_m}, both in 
Appendix~\ref{subsec:qppt}.

The case where both Adversary and Simulator are $\QPPT$ is called 
\emph{classical-style} security (as it is simply a weaker variant of 
Stand-Alone Security in the usual sense of \cite{HSS15}), while a protocol 
that remains secure when both are $\BQP$ is said to be 
\emph{superposition-resistant}. This allows us to demonstrate a separation 
between Adversaries with and without superposition access 
(machines in $\QPPT$ have the same computing power as $\BQP$-machines but 
operate solely on classical input and output data). 
Note that $\QPPT$-machines can be seen as restricted $\BQP$-machines and so 
superposition-resistance implies classical-style security.

Quantifying Definition \ref{sec-def} over a subset of Adversaries in each class yields 
flavours such as Honest-but-Curious or Malicious. The behaviour of an 
Honest-but-Curious $\QPPT$ Adversary is the same as a classical 
Honest-but-Curious Adversary during the protocol but it may use its quantum 
capabilities in the post-processing phase of its attack. We define an 
extension of these Adversaries in Definition \ref{ext-hbc}: they are almost 
Honest-but-Curious in that there is an Honest-but-Curious Adversary whose 
Simulator also works for the initial Adversary (therefore satisfying the security Definition \ref{sec-def}). This is required as the
adversarial behaviour of our attack is not strictly Honest-but-Curious when
translated to classical messages, but it does follow this new definition.

\begin{definition}[Extended Honest-but-Curious Adversaries]
\label{ext-hbc}
Let $\Pi$ be a protocol that is secure according to Definition \ref{sec-def} 
against Honest-but-Curious $\QPPT$ Adversaries. We say that an Adversary 
$\mathcal{A}$ is Extended Honest-but-Curious if there exists an 
Honest-but-Curious Adversary $\mathcal{A}'$ such that the associated Simulator 
$\mathcal{S}'$ satisfies Definition \ref{sec-def} for $\mathcal{A}$ if we 
allow it to output $\Abort$ when the honest party would abort as well.
\end{definition}


\headingb{Ideal Functionality Behaviour and Formal Security Definition.} 
This section differs crucially from previous models of security. The Two-Party Computation 
Ideal Functionality implementing a binary function 
$f$, formally defined as Ideal Functionality \ref{IdealF2PC} (Appendix 
\ref{app:OTP}), takes as input a quantum state from each party,
measures it in the computational basis, applies the function $f$ to the 
classical measurement results and returns the classical inputs to each party 
while one of them also receives the output. 
\footnote{This is wlog.\ classically, see Appendix \ref{app:OTP} and Section \ref{subsec:path}.}

While it can seem highly counter-intuitive to consider an ideal scenario where a 
measurement is performed (since it is not present in the real scenario),
this measurement by the Ideal Functionality is necessary in order to have a 
meaningful definition of security. It is only if the protocol with superposition access
behaves similarly to a classical protocol that 
it can be considered as resistant to superposition attacks. 
It is therefore precisely because we wish to capture the security against superposition attack, 
that we define the Ideal Functionality as purely classical (hence the 
measurement). If the Ideal Adversary (a Simulator interacting classically with the Ideal 
Functionality) and the Real Adversary (which can interact in superposition 
with the honest player) are indistinguishable, 
only then is the protocol superposition-secure.

Furthermore, as argued briefly in the Introduction, 
Ideal Functionalities which do not measure the inputs of both parties when 
they receive them as they always 
allow superposition attacks, which then extract more information than the classical 
case (as proven in \cite{SSS15}). A superposition attack against a protocol 
implementing such a functionality is therefore not considered an attack since it is by definition a 
tolerated behaviour in the ideal scenario.

We can now give our security Definition \ref{sec-def}. A protocol between parties $\Pone$ and $\Ptwo$ is said to securely compute two-party functions of a given set $\mathfrak{F}$ against corrupted party $\Pone^*$ if, for all functions $f : \{0, 1\}^{n_X} \times \{0, 1\}^{n_Y} \longrightarrow \{0, 1\}^{n_Z}$ with $ f \in \mathfrak{F}$, 
no Adversary controlling $\Pone^*$ can distinguish between the real and ideal executions with high probability. 

\begin{definition}[Computational Security against Adversary Class 
$\mathfrak{X}$]
\label{sec-def}
Let $\epsilon(\eta) = o(1)$ be a function of the security parameter $\eta$. 
Let $f \in \mathfrak{F}$ be the function to be computed by protocol $\Pi$ 
between parties $\Pone$ and $\Ptwo$.
We say that a protocol~$\Pi$ $\epsilon(\eta)$-securely emulates Ideal 
Functionality $\mathcal{F}$ computing functions from set~$\mathfrak{F}$ 
against $\mathfrak{X}$-adversarial $\Pone^*$ (with $\mathfrak{X} \in \{\QPPT, 
\BQP\}$) if for all Adversaries $\mathcal{A}$ in class $\mathfrak{X}$ 
controlling the corrupted party $\Pone^*$ and all quantum polynomial-time 
Environments $\mathcal{Z}$,
there exists a Simulator $S_{\Pone^*}$ in class $\mathfrak{X}$ such that:
\[
\Bigl\lvert \mathbb{P}\Bigl[b = 0 \mid b \leftarrow \mathcal{A}\Bigl(v_{\mathcal{A}}(S_{\Pone^*}, \rho_{\mathcal{A}})\Bigr)\Bigr] -  \mathbb{P}\Bigl[b = 0 \mid b \leftarrow \mathcal{A}\Bigl(v_{\mathcal{A}}(\Ptwo, \rho_{\mathcal{A}})\Bigl)\Bigr] \Bigr\rvert \leq \epsilon(\eta)
\]


In the equation above, the variable $v_{\mathcal{A}}(S_{\Pone^*}, 
\rho_{\mathcal{A}})$ corresponds to the final state (or view) of the Adversary 
in the ideal execution when interacting with Simulator~$S_{\Pone^*}$ with 
Ideal Functionality $\mathcal{F}$ and $v_{\mathcal{A}}(\Ptwo, 
\rho_{\mathcal{A}})$ corresponds to the final state of the Adversary when 
interacting with honest party $\Ptwo$ in the real protocol $\Pi$. The 
probability is taken over all executions of protocol $\Pi$.
\end{definition}

In the case where one party does not receive an output, it is possible to reduce the security property to input-indistinguishability, defined below in Definition \ref{inp-ind}.

\begin{definition}[Input-Indistinguishability]
\label{inp-ind}
Let $\Pi$ be protocol between parties $\Pone$ and $\Ptwo$ with input space 
$\{0, 1\}^{n_Y}$ for $\Ptwo$. We say that the execution of~$\Pi$ is 
\emph{$\epsilon$-input-indistinguishable} for $\Pone^*$ if there exists an 
$\epsilon(\eta) = o(1)$ such that, for all computationally-bounded quantum 
Distinguishers $\mathcal{D}$ and any two inputs $y_1, y_2 \in \{0, 1\}^{n_Y}$:
\[
\Bigl\lvert \mathbb{P}\Bigl[b = 0 \mid b \leftarrow \mathcal{D}\Bigl(v_{\mathcal{A}}(\Ptwo(y_1), \rho_{\mathcal{A}})\Bigr)\Bigr] -  \mathbb{P}\Bigl[b = 0 \mid b \leftarrow \mathcal{D}\Bigl(v_{\mathcal{A}}(\Ptwo(y_2), \rho_{\mathcal{A}})\Bigl)\Bigr] \Bigr\rvert \leq \epsilon(\eta)
\]
In the equation above, the variable $v_{\mathcal{A}}(\Ptwo(y_i), 
\rho_{\mathcal{A}})$ corresponds to the final state of the Adversary when 
interacting with honest party $\Ptwo$ (with input $y_i$) in the real protocol 
$\Pi$. The probability is taken over all executions of protocol $\Pi$.
\end{definition}

We can now state Lemma \ref{ind-to-sec} (its proof can be found in 
Appendix \ref{app:OTP}).


\begin{lemma}[Input-Indistinguishability to Security]
\label{ind-to-sec}
\newcounter{compteur:ind-to-sec}
\setcounterref{compteur:ind-to-sec}{ind-to-sec}
Let $f \in \mathfrak{F}$ be the function to be computed by protocol $\Pi$ 
between parties $\Pone$ and $\Ptwo$, where $\mathfrak{F}$ is the set of 
functions taking as input $(x, y) \in \{0, 1\}^{n_Y} \times \{0, 1\}^{n_X}$ 
and outputting $z \in \{0, 1\}^{n_Z}$ to $\Ptwo$ (and no output to $P_1$). If 
the protocol is input-indistinguishable for adversarial $\Pone^*$ in class 
$\mathfrak{X}$ (Definition \ref{inp-ind}) then it is secure against 
adversarial $\Pone^*$ in class $\mathfrak{X}$ (Definition \ref{sec-def}) with 
identical bounds.
\end{lemma}


\headingb{Comments on the Security Model.} 
We show that Definition~\ref{sec-def} is achievable by giving a proof that the 
Classical One-Time Pad is secure against superposition attacks (see Appendix 
\ref{app:OTP}). 
In our security model, both the Adversary and Simulator can have superpositions of states as input. The only differences is that, if the Simulator chooses to send a state to the Ideal Functionality, it knows that this third party will perform on it a measurement in the computational basis. 
Note that in any security proof, 
the Simulator may choose not to perform the call to the Ideal Functionality. 
This is because the security definition does not force the Simulator to 
reproduce faithfully the output of the honest Client, as the distinguishing 
is done only by the Adversary and not a global distinguisher as in 
\cite{HSS15}. 
This also means that sequential composability 
explicitly does not hold with such a definition, even with the most basic 
functionalities (whereas the Stand-Alone Framework of \cite{HSS15} guarantees 
it). An interesting research direction would be to find a composable 
framework for proving security against superposition attacks and we leave 
this as an open question.

%% file: tex_files/04a_yao_protocol.tex
In order to demonstrate the capabilities of our new model in the case of more 
complex two-party scenarios, we will analyse the security of the well-known 
Yao Protocol, the pioneer in Secure Two-Party Computation, against $\QPPT$ and 
$\BQP$ Adversaries.

Its purpose is to allow two Parties, the Garbler and the Evaluator, to compute a joint 
function on their two classical inputs. The Garbler starts by preparing an encrypted 
version of the function and then the Evaluator decrypts it using keys that correspond 
to the two players' inputs, the resulting decrypted value being the final output.


The Original Yao Protocol secure against Honest-but-Curious classical Adversaries 
has first been described by Yao in the oral presentation for 
\cite{Yao86}, but a rigorous formal proof was only presented in 
\cite{LP04proof}. It has been proven secure against quantum Adversaries
with no superposition access to the honest player in \cite{BDK20} (for a 
quantum version of IND-CPA that only allows random oracle query to be in 
superposition).

We start by presenting informal definitions for symmetric encryption schemes in 
Section \ref{subsec:sym_enc_def}
(the formal definitions are presented in 
Appendix \ref{subsec:symmetric:encryption}).
We then present in Section \ref{orig_yao} the garbled table 
construction which is the main building block of Yao's Protocol and give an 
informal description of the Original Yao Protocol. Then in Section 
\ref{prot_prez} we give a description of a slight variant of the 
original protocol, resulting in the Modified Yao Protocol. The proofs of correctness and 
$\QPPT$-security are given in Appendix~\ref{proofs_CS} and show that the modifications 
do not make the protocol less secure in the classical case or against $\QPPT$ 
Adversaries, but will make superposition attacks possible as presented 
in Section~\ref{sec:attack}.


\subsection{Definitions For Symmetric Encryption Schemes}
\label{subsec:sym_enc_def}

An encryption scheme consists of two classical efficiently computable 
deterministic functions $\Enc : \mathfrak{K} \times \mathfrak{A} \times 
\mathfrak{M} \rightarrow \mathfrak{K} \times \mathfrak{A} \times \mathfrak{C}$ 
and $\Dec : \mathfrak{K} \times \mathfrak{A} \times \mathfrak{C} \rightarrow 
\mathfrak{K} \times \mathfrak{A} \times \mathfrak{M}$ (where $\mathfrak{K}$ is 
the set of valid keys, $\mathfrak{A}$ the set of auxiliary inputs, 
$\mathfrak{M}$ the set of plaintext messages and $\mathfrak{C}$ the set of 
ciphertexts, which is supposed equal to~$\mathfrak{M}$). 
We suppose that for all $(k, \mathrm{aux}, m) \in \mathfrak{K} \times \mathfrak{A} \times 
\mathfrak{M}$, we have that $\Dec_k(\mathrm{aux}, \Enc_k(\mathrm{aux}, m)) = m$.

We will use a symmetric encryption scheme with slightly different properties compared to the original protocol of \cite{Yao86} or \cite{LP04proof}. The purpose of these modifications is to make it possible to later represent the action of the honest player (the decryption of garbled values) using a minimal oracle representation when embedded as a quantum protocol (as described in Lemma \ref{min_or}). 
We give in Appendix \ref{subsec:symmetric:encryption} sufficient conditions implying this definition and a concrete instantiation of a symmetric encryption scheme that satisfies them.


\begin{definition}[Minimal Oracle Representation]
\label{def_min_or}
Let $(\Enc, \Dec)$ be an encryption scheme defined as above, we say that it has a Minimal Oracle Representation if there exists efficiently computable unitaries $\mathsf{M}_{\Enc}$ and $\mathsf{M}_{\Dec}$, called minimal oracles, such that for all $k \in \mathfrak{K}$, $\mathrm{aux} \in \mathfrak{A}$ and $m \in \mathfrak{M}$, $\mathsf{M}_{\Enc}\ket{k}\ket{\mathrm{aux}}\ket{m} = \ket{e_K(k)}\ket{e_A(\mathrm{aux})}\ket{\Enc_k(\mathrm{aux}, m)}$ (in which case $\mathsf{M}_{\Enc}^{\dagger} = \mathsf{M}_{\Dec}$), where $e_K$ and $e_A$ are efficiently invertible permutations of the key and auxiliary value.
\end{definition}

The requirement above forces us to define the symmetric encryption scheme as 
secure if it is a quantum-secure pseudo-random permutation. We give an 
informal definition, the formalised version can be found as Definition 
\ref{sec_sym} in Appendix \ref{subsec:symmetric:encryption}. 
For a discussion on this choice of security definitions, see Appendix  \ref{subsec:comments}.


\begin{definition}[Real-or-Permutation Security of Symmetric Encryption (Informal)]
\label{inf_sec_sym}
A symmetric encryption scheme is said to be secure against quantum Adversaries if the distinguishing advantage of a computationally-bounded quantum Adversary in the following game is negligible in the security parameter:
\begin{enumerate}
\item The Challenger chooses either uniformly at random a permutation over the plaintext message space $\mathfrak{M}$ or samples a key uniformly at random from $\mathfrak{K}$.
\item For all encryption queries, the Adversary sends a state $\rho_i$ of its choice to the Challenger.
\item The Challenger responds by applying the minimal oracle defined in the first step and sending the result back to the Adversary. 
\item The Adversary guesses whether it interacted with the real encryption function or a random permutation.
\end{enumerate}

\end{definition}

\subsection{The Original Yao Protocol}
\label{orig_yao}

The protocol will be presented in a hybrid model where both players have access to a trusted third party implementing a 1-out-of-2 String Oblivious Transfer (Ideal Functionality \ref{IdealOT}). The Garbler plays the role of the Sender of the OT while the Evaluator is the Receiver. The attack presented further below does not rely on an insecurity from the OT, which will be supposed to be perfectly implemented and, as all Ideal Functionalities in this model, without superposition access. 
As a consequence of the non-composability of our framework, replacing an Ideal Functionality with a protocol that constructs it does not guarantee that the global construction is secure. It deserves to be noted that our attack does not rely on anything but the classical correctness of the Oblivious Transfer, so this is not relevant to our study.

We focus on the case where the output is a single bit. Suppose that the 
Garbler and Evaluator have agreed on the binary function to be evaluated $f : 
\{0, 1\}^{n_X} \times \{0, 1\}^{n_Y} \longrightarrow \{0, 1\}$, with the 
Garbler's input being $x \in \{0, 1\}^{n_X}$ and the Evaluator's input being 
$y \in \{0, 1\}^{n_Y}$. 
The protocol can be summarised as follows. 
The Garbler $G$ generates a garbled circuit $\mathit{GC}_f$ (defined 
below), along with keys $\qty{k_0^{G, i}, k_1^{G, i}}_{i \in [n_X]}$ and 
$\qty{k_0^{E, i}, k_1^{E, i}}_{i \in [n_Y]}$ for the Garbler's and Evaluator's input respectively. 
To each bit of input correspond 
two keys, one (lower-indexed with $0$) if the player chooses the value $0$ for 
this bit-input and the other if it chooses the value $1$. They invoke $n_Y$ 
instances of a 1-out-of-2 String OT Ideal Functionality, the Evaluator's 
input (as Receiver of the OT) to these is $y_i$ for $i \in [n_Y]$, while the 
Garbler inputs (as Sender) the keys $(k_0^{E, i}, k_1^{E, i})$ corresponding 
to input $i$ of the Evaluator. The Evaluator therefore recovers $k_{y_i}^{E, 
i}$ at the end of each activation of the OT. The Garbler then sends the keys 
$\qty{k_{x_i}^{G, i}}_{i \in [n_X]}$ corresponding to its own input along with 
the garbled circuit which is constructed as follows.


Let $(\Enc, \Dec)$ be a symmetric encryption scheme. To construct the garbled table for a gate computing a two-bit function $g$, with inputs wires labelled $a$ and $b$ and output wire $z$, the Garbler first chooses keys $(k_0^a, k_1^a, k_0^b, k_1^b) \in \mathfrak{K}^4$ for the input wires and $k^z \in \{0, 1\}$ for the output\footnote{The value $k^z$ is used to One-Time-Pad the outputs, preserving security for the Garbler after decryption as only one value from the garbled table can be decrypted correctly.}. Let $\mathrm{aux}_a$ and $\mathrm{aux}_b$ be two auxiliary values for the encryption scheme. It then iterates over all possible values $\tilde{a}, \tilde{b} \in \{0, 1\}$ to compute the garbled table values $E_{\tilde{a}, \tilde{b}}^{k^z}$ defined as (with padding length $p = n_M - 1$, where $n_M$ is the bit-length of the messages of the encryption scheme and $\parallel$ represents string concatenation):


\[
E_{\tilde{a}, \tilde{b}}^{k^z} := \Enc_{k_{\tilde{a}}^a}\Bigl(\mathrm{aux}_a, \Enc_{k_{\tilde{b}}^b}(\mathrm{aux}_b, g(\tilde{a}, \tilde{b}) \oplus k^z \parallel 0^p)\Bigr) \\
\]


The ordered list thus obtained is called the \emph{initial} garbled table. The Garbler then chooses a random permutation $\pi \in \mathcal{S}_4$ and applies it to this list, yielding the \emph{final} garbled table $\mathit{GT}_{g}^{(a, b, z)}$. 
For gates with fan-in $l$, the only difference is that the number of values in the table will be $2^l$, the rest may be computed in a similar way (by iterating over all possible values of the function's inputs). The keys are always used in an fixed order which is known to both players at time of execution (we suppose for example that, during encryption, all the keys of the Evaluator are applied first, followed by the keys of the Garbler).

Finally, after receiving the keys (through the OT protocols for its own, 
and via direct communication for the Garbler's) and garbled table, the 
Evaluator uses them to decrypt sequentially each entry of the 
table and considers it a success if the last $p$ bits are equal to $0$ 
(except with probability negligible in $p$, the decryption of 
a ciphertext with the wrong keys will not yield $p$ bits set to $0$, see Lemma \ref{prot_corr}). 
It~then returns the corresponding register to the Garbler.

\subsection{Presentation of the Modified Yao Protocol}
\label{prot_prez}

\paragraph{Differences with the Original Yao Protocol.} There are four 
main differences between our Modified Yao Protocol \ref{Yao} and the 
well-known protocol from \cite{Yao86} recalled above. The first two are 
trivially just as secure in the classical case (as they give no more power to 
either player): the Garbler sends one copy of its keys to the Evaluator for 
each entry in the garbled table and instructs it to use a ``fresh" copy for 
each decryption; and the Evaluator returns to the Garbler the copy of the 
Garbler's keys that were used in the successful decryption. Notice also that 
there is only one garbled table for the whole function instead of a series of 
garbled tables corresponding to gates in the function's decomposition. This is 
less efficient but no less secure than the original design in the classical 
case (and quantum case without superposition access), as a player breaking the 
scheme for this configuration would only have more power if it has access to 
intermediate keys as well. The last difference is the use of a weaker security 
assumption for the symmetric encryption function (indistinguishability from a 
random permutation instead of the quantum equivalents to IND-CPA security 
developed in \cite{BZ13,GHS16,MS16}). This lower security 
requirement is imposed in order to model the honest 
player's actions using the minimal oracle representation. This property 
influences the security against an adversarial Evaluator, but Theorem \ref{sec_yao} 
shows that this assumption is sufficient for security in our scenario. 
The reasons for these modifications, related to our attack, are developed in Appendix \ref{subsec:comments}.

The full protocol for a single bit of output is described in Protocol~\ref{Yao}. 
The correctness and security against $\QPPT$-Adversaries of this Modified Yao Protocol are captured by Theorems \ref{prot_corr} and \ref{sec_yao} (see Appendix \ref{proofs_CS} for proofs), showing that the modifications above have no impact against these Adversaries.


\begin{theorem}[Correctness of the Modified Yao Protocol]
\label{prot_corr}
\newcounter{compteur:prot_corr}
\setcounterref{compteur:prot_corr}{prot_corr}
Let $(\Enc, \Dec)$ be a symmetric encryption scheme with a Minimal Oracle Representation 
(Definition \ref{def_min_or}). Protocol \ref{Yao} is correct with probability exponentially close to $1$ in 
$\eta$ for $p = \mathit{poly}(\eta)$.
\end{theorem}


\begin{theorem}[$\QPPT$-Security of the Modified Yao Protocol]
\label{sec_yao}
\newcounter{compteur:sec_yao}
\setcounterref{compteur:sec_yao}{sec_yao}
Consider a hybrid execution where the Oblivious Transfer is handled by a 
classical trusted third party.
Let $(\Enc, \Dec)$ be a symmetric encryption scheme that is $\epsilon_{Sym}$-real-or-permutation-secure 
(Definition \ref{sec_sym}). Then Protocol 
\ref{Yao} is perfectly-secure against a $\QPPT$ adversarial Garbler (the 
Adversary's advantage is $0$) and $(2^{n_X + n_Y} - 1) 
\epsilon_{Sym}$-secure against $\QPPT$ adversarial Evaluator according to 
Definition \ref{sec-def}.
\end{theorem}

%% file: tex_files/05_attack_yao.tex
Section \ref{state_gen} presents a superposition attack on the Modified Yao Protocol (Protocol~\ref{Yao}). 
The formal version of the Attacks are found in Appendix~\ref{app:form_att} and the proofs of the accompanying Theorems in Appendix~\ref{app:proof_insec}. 
The attack is further optimised in Appendix \ref{sec:OT} using the free-XOR 
technique, and applied to an instance of Yao's Protocol computing an 
Oblivious Transfer. 
Section~\ref{subsec:path} then analyses it post-mortem to build a Superposition-Resistant Yao Protocol.

Note that this attack does not simply distinguish between the ideal and real executions, but allows the Adversary to extract one bit of information from the honest player's input. It is therefore a concrete attack on the Modified Yao Protocol~\ref{Yao} (as opposed to a weaker statement about not being able to perform an indistinguishable simulation in our model).

\subsection{Attacking the Modified Yao Protocol via Superpositions}
\label{state_gen}

In the following, the classical protocol is embedded in a quantum framework, all message are stored in quantum registers as quantum states that can be in superposition. The encryption and decryption procedures are performed using the Minimal Oracle Representation from Definition \ref{def_min_or}. The OT Ideal Functionality \ref{IdealOT} from Appendix \ref{app:OTP} measures the inputs and outputs states in the computational basis. The checks of the Evaluator on the padding for successful decryption are modelled as a quantum measurement of the corresponding register.

We start by presenting the action of the adversarial Garbler during the execution of Protocol~\ref{Yao} (its later actions are described below). Its aim is to generate a state containing a superposition of its inputs and the corresponding outputs for a fixed value of the Evaluator's input. This State Generation Procedure on the Modified Yao Protocol \ref{Yao} (Attack \ref{Yao_att}) can be summarised as follows (see Theorem \ref{correct_att} for its analysis):
\begin{enumerate}
\item The Adversary's choice of keys, garbled table generation (but for both values of $k^z$) and actions in the OT are performed honestly.
\item Instead of sending one set of keys as its input, it sends a superposition of keys for two different non-trivial values of the Garbler's input $(\widehat{x_0}, \widehat{x_1})$ (they do not uniquely determine the output).
\item For each value in the garbled table, it instead sends a uniform superposition over all calculated values (with a phase of $-1$ for states representing garbled values where $k^z = 1$).
\item It then waits for the Evaluator to perform the decryption procedure and, if the Evaluator succeeded in decrypting one of the garbled values and returns the output and register containing the Garbler's keys, the Adversary performs a clean-up procedure which translates each key for bit-input $0$ (respectively $1$) into a logical encoding of $0$ (respectively $1$). This procedure depends only on its own choice of keys.
\end{enumerate}
\begin{theorem}[State Generation Analysis]
\label{correct_att}
\newcounter{compteur:correct_att}
\setcounterref{compteur:correct_att}{correct_att}
The state contained in the Garbler's attack registers at the end of a 
successful Superposition Generation Procedure (Attack \ref{Yao_att}) is negligibly close to 
$\frac{1}{2}\sum\limits_{x, k^z} (-1)^{k^z}\ket{x^L}\ket{f(x, \hat{y}) \oplus 
k^z}$, where $x^L$ is a logical encoding of $x$ and $x \in \{\widehat{x_0}, 
\widehat{x_1}\}$. Its success probability is lower bounded by $1 - e^{-1}$ for all values of $n_X$ 
and $n_Y$.
\end{theorem}


We show also in Appendix \ref{app:form_att} that, if $U_f^{\hat{y}}$ is the Standard Oracle applying function $f(\cdot, \hat{y})$ (ie.\ $U_f^{\hat{y}}\ket{x}\ket{k^z} = \ket{x}\ket{f(x, \hat{y}) \oplus k^z}$), then it is possible to generate $U_f^{\hat{y}}\ket{\psi}\ket{\phi}$ for any states $\ket{\psi}$ (over $n_X$ qubits) and $\ket{\phi}$ (over one qubit) with efficient classical descriptions by using the same technique.

 
We can now analyse the actions of the Adversary after the protocol has terminated. 
The Full Attack \ref{Yao_full_att} breaking the security of the Modified Yao Protocol \ref{Yao} (Theorem \ref{sup_att_analysis}, proof in Appendix \ref{app:proof_insec}) can be summarised as follows:
\begin{enumerate}
\item The Environment provides the Adversary with the values of the Garbler's input $(\widehat{x_0}, \widehat{x_1})$. The input of the honest Evaluator is $\hat{y}$.
\item The Adversary performs the State Generation Procedure with these inputs.
\item If it has terminated successfully, the Adversary performs an additional clean-up procedure (which only depends on the values of $(\widehat{x_0}, \widehat{x_1})$) to change the logical encoding of $\widehat{x_b}$ into an encoding of $b$. The resulting state is (omitting this logical encoding, with $b_i := f(\widehat{x_i}, \hat{y})$ and up to a global phase):
\[
\frac{1}{\sqrt{2}}\bigl(\ket{0} + (-1)^{b_0 \oplus b_1}\ket{1}\bigr)\otimes\ket{-}
\]
\item The Adversary applies the final steps of the DJ algorithm (after the application of the oracle, see Appendix \ref{app:more_q}) to recover the XOR of the output values for the two inputs: it applies a Hadamard gate to its first register and measures it in the computational basis. 
\end{enumerate}

\begin{theorem}[Vulnerability to Superposition Attacks of the Modified Yao Protocol]
\label{sup_att_analysis}
\newcounter{compteur:sup_att_analysis}
\setcounterref{compteur:sup_att_analysis}{sup_att_analysis}
For any non-trivial two-party function $f : \{0, 1\}^{n_X} \times \{0, 
1\}^{n_Y} \rightarrow \{0, 1\}$, let $(\widehat{x_0}, \widehat{x_1})$ be a 
pair of non-trivial values in $\{0, 1\}^{n_X}$. 
For all inputs $\hat{y}$ of honest Evaluator in 
Protocol \ref{Yao}, let $\PredE(\hat{y}) = f(\widehat{x_0}, \hat{y}) \oplus 
f(\widehat{x_1}, \hat{y})$. Then there exists a real-world $\BQP$ Adversary 
$\mathcal{A}$ against Protocol \ref{Yao} implementing $f$ such that for any 
$\BQP$ Simulator $\mathcal{S}$, the advantage of the Adversary over the 
Simulator in guessing the value of $\PredE(\hat{y})$ is lower-bounded by 
$\frac{1}{2}(1 - e^{-1})$.
\end{theorem}

If the ideal and real executions were indistinguishable according to 
Definition~\ref{sec-def}, such a feat would be impossible for the Adversary 
since the Simulator can at most access one value of the output through the 
Ideal Functionality.
%

Finally, the following lemma captures the fact that the 
previously described Adversary does not break the Honest-but-Curious security of the 
Modified Yao Protocol if it does not have superposition access (a 
fully-malicious one can trivially break it), thereby demonstrating the 
separation between Adversaries with and without superposition access 
(see Appendix~\ref{app:form_att} for the proof).

\begin{lemma}[Adversarial Behaviour Analysis]
\label{lem:adv_beh}
\newcounter{compteur:lem:adv_beh}
\setcounterref{compteur:lem:adv_beh}{lem:adv_beh}
The \emph{$\QPPT$-reduced machine} (Definition \ref{qppt_m}) corresponding to 
the $\BQP$ Adversary described in Attack~\ref{Yao_full_att} is an \emph{Extended 
Honest-but-Curious Adversary} (Definition \ref{ext-hbc}).
\end{lemma}

\subsection{Superposition-Resistant Yao Protocol}
\label{subsec:path}

We can now analyse the crucial points where the security 
breaks down and propose counter-measures. 
We notice that all actions of the Adversary only act on 
the registers that contain its own keys (recall that the Evaluator 
sends back the Garbler's keys after a successful decryption) and have no effect on the output 
register, which stays in the $\ket{-}$ state the whole time. It is thus 
unentangled from the rest of the state and the attack on the protocol can therefore also 
be performed if the Garbler has no output. As the security in this case 
still holds for $\QPPT$ Adversaries via input-indistinguishability,
it means that this security property does not carry over from $\QPPT$ to $\BQP$ either.

Therefore, as counter-intuitive as it may seem, the precise point that makes the attack possible 
is a seemingly innocuous message consisting of 
information that the Adversary should (classically) already have, 
along with a partial measurement on the part of the honest player (which is even 
stranger considering that it is usually thought that the easiest way to 
prevent superposition attack is to measure the state).

Not sending back this register to the 
Adversary (as in the Original Yao Protocol) makes the protocol 
structurally similar to the One-Time-Pad Protocol~\ref{otp}: one 
party sends everything to the other, who then simply applies local 
operations. The proof for the One-Time-Pad works by showing that there is a 
violation of the no-signalling condition if the Adversary is able to guess whether it is in the real or ideal situation. 
This technique can be reused if the Evaluator does not give 
away the result of the measurement on its state (by hiding the success or failure of the garbled table decryption\footnote{This contradicts the remark in Appendix \ref{app:IF} after Ideal Functionality \ref{IdealF2PC} since the proof works if there is no future communication between the two players.}). This Superposition-Resistant Yao Protocol \ref{Yao_sup} and proof of Theorem \ref{Yao_sup_sec}\footnote{As noted in Section \ref{sec:security:model}, superposition-resistance implies classical-style security.} are described in Appendix \ref{app:sup_sec_yao}.

\begin{theorem}[$\BQP$-Security of Superposition-Resistant Yao Protocol]
\label{Yao_sup_sec}
\label{Yao-sup-sec}
\newcounter{compteur:Yao-sup-sec}
\setcounterref{compteur:Yao-sup-sec}{Yao-sup-sec}
The Superposition-Resistant Yao Protocol \ref{Yao_sup} is perfectly-secure against a 
$\BQP$ adversarial Garbler according to Definition \ref{sec-def} in an OT-hybrid execution.
\end{theorem}
%
%
%
%

%% file: tex_files/07_conclusion.tex
Our security model and the attack analysis performed in 
this paper lie completely outside of the existing models of security against superposition attacks. They either consider the computational security of basic primitives or, for more complex protocols with multiple interactions between distrustful parties, the protocols are all considered to be statistically-secure (and are therefore essentially extensions of \cite{commit2}). This leads to many simplifications which have no equivalent in the computational setting.
We develop a novel security framework, based on the simple premise that to be secure from superposition attacks means emulating a purely classical functionality. We show that, given slight modifications that preserves classical security, it is possible to show superposition attacks on computationally-secure protocols. The intuition gained from the attack allows us to build a computationally superposition-resistant protocol for Two-Party Secure Function Evaluation, a task never achieved before.

Our results demonstrate once again the counter-intuitive nature of 
quantum effects, regarding not only the 
vulnerability of real-world protocols to superposition attacks (most would 
require heavy modifications for known attacks to work), 
but also attack vectors and the optimal ways to counter them (as partial 
measurements can even lead to attacks).
%

%% file: tex_files/08_ack.tex
This work was supported in part by the French ANR project CryptiQ (ANR-18-CE39-0015). 
We acknowledge support of the European Union’s Horizon 2020 Research and 
Innovation Program under Grant Agreement No. 820445 (QIA).
We would like to thank Michele Minelli, Marc Kaplan and Ehsan Ebrahimi 
for fruitful discussions.

%% file: tex_files/A0_more_q.tex
We give here a brief overview of quantum systems and a few basic operations and refer to \cite{nielsenchuang} for a more detailed presentation.

Any pure quantum state is represented by a vector $\ket{\psi}$ in a given 
Hilbert space $\mathcal{H}$, which in the simplest case is $\mathbb{C}^2$ for qubits (which will always be the case in this paper). For $n$ qubits, the joint system is given by $\mathbb{C}^{2^n} = \mathbb{C}^2 \otimes \ldots \otimes \mathbb{C}^2$ for $n$ subspaces, where $\otimes$ designates the tensor product of Hilbert spaces. We will use the term quantum register in the same sense as a classical memory register in a classical computer (as a way to reference specific qubits or subsystems). For $n$ qubits, we call computational basis the family of classical bit-string states $\mathcal{B}_{C} = \qty{\ket{x} \mid x \in \{0,1\}^n}$. Let $\1_{\mathcal{A}}$ be the identity operation on quantum register $\mathcal{A}$. We write $\dagger$ for the conjugate transpose operation and $\bra{\phi} = \ket{\phi}^{\dagger}$. This in turn gives us the inner-product $\braket{\phi}{\psi}$ and projector $\dyad{\psi}$.

More generally, if the quantum state is in pure state $\ket{\phi_i}$ with probability $p_i$ then the system is described as the density matrix $\rho = \sum\limits_i p_i \dyad{\phi_i}$ (also called mixed state). Let $D(\mathcal{Q})$ be the set of all possible quantum states in a given quantum register $\mathcal{A}$: it is the set of all Hermitian mixed states with trace equal to $1$ and positive eigenvalues. In general, the input to a protocol is a mixed state $\rho_{in} \in D(\Regone\otimes\Regtwo\otimes\mathcal{W})$, where $\mathcal{W}$ is an auxiliary register (the inputs are potentially entangled to this reference register).

Unitaries acting on register $\mathcal{Q}$ are linear operations $U$ such that $U^{\dagger} U = \Id_{\mathcal{Q}}$. On the other hand, a measurement on a quantum register $\mathcal{Q}$ is represented in the simplest case (which will be sufficient here) by a complete set of orthogonal projectors $\qty{P_m}$ satisfying $\sum\limits_m P_m = \Id_{\mathcal{Q}}$ and $P_m P_m' = \delta_{m, m'} P_m$, where $\delta_{m, m'}$ is Kronecker's delta. Then the probability of obtaining output $m$ by the measurement defined above on state $\ket{\psi}$ is given by $p(m) = \ev{P_m}{\psi}$, the post-measurement state is then $\frac{P_m \ket{\psi}}{\sqrt{p(m)}}$.

Let $L(\mathcal{A})$ be the set of linear mappings from $\mathcal{A}$ to 
itself. If $\mathcal{E} : L(\mathcal{A}) \rightarrow L(\mathcal{B})$ is a completely positive and trace non-decreasing superoperator, it is called \emph{quantum operation} or CPTP-map. It can always be decomposed into unitaries followed by measurements in the computational basis. For any quantum register $\mathcal{Q}$ and any state $\rho_{Q}$ is it always possible to define, given another sufficiently large quantum system $\mathcal{R}$, a pure state $\ket{\phi_{RQ}}$ such that looking at the restriction of the system to register $\mathcal{Q}$ (by tracing out $\mathcal{R}$) gives $\rho_Q$. This technique is called purification, the register $\mathcal{R}$ is called the reference or ancillary register, and allows to represent any CPTP-map as a unitary on a larger system.

We can now give some standard quantum operations used throughout the paper. The Pauli $\X$ operator is defined as $\X = \begin{pmatrix} 0 &&& 1 \\ 1 &&& 0 \end{pmatrix}$ (corresponding to a bit-flip classically), while the $\CNOT$ gate (with the first qubit being the control) is defined through $\CNOT\ket{0}\ket{\phi} = \ket{0}\ket{\phi}$ and $\CNOT\ket{1}\ket{\phi} = \ket{1}\X\ket{\phi}$ for any state $\ket{\phi}$. The Pauli $\Z$ operator is defined as $\Z = \begin{pmatrix} 1 &&& 0 \\ 0 &&& -1 \end{pmatrix}$. A logical Hadamard gate $\Ha^L$ is defined by $\Ha_L\ket{0}^{\otimes L} = \ket{+_L} = \frac{1}{\sqrt{2}}\Bigl(\ket{0}^{\otimes L} + \ket{1}^{\otimes L}\Bigr)$, $\Ha_L\ket{1}^{\otimes L} = \ket{-_L} = \frac{1}{\sqrt{2}}\Bigl(\ket{0}^{\otimes L} - \ket{1}^{\otimes L}\Bigr)$ ($\Ha_L$ acts as identity on the remaining basis states). The state $\ket{+_L}$ is more commonly referred to as the $\ket{\GHZ_L}$ state, and $\ket{-_L} = \Z_1\ket{\GHZ_L}$.

It is possible to represent any classical operation using a quantum implementation of the reversible classical Toffoli gate computing the function $T(a, b, c) = (a \cdot b) \oplus c$ where $(\oplus, \cdot)$ are defined in $\mathbb{Z}_2$. This can be defined as a unitary on three qubits (any reversible classical gate is simply a permutation of the computational basis states) and is universal for classical computations. Any binary function $f: \{0, 1\}^n \rightarrow \{0, 1\}^m$ can therefore be implemented as a unitary $U_f$ defined on computational basis states $\ket{x}\ket{y}$ (with $x \in \{0, 1\}^n$ and $y \in \{0, 1\}^m$) as $U_f\ket{x}\ket{y} = \ket{x}\ket{y \oplus f(x)}$ (called standard oracle of $f$).

Since our attack resembles in spirit the Deutsch-Jozsa algorithm, we recall here the principle. The point of this algorithm is to solve the following promise problem: given a function $f$ outputting a single bit, determine whether it is constant (the output bit is the same for all inputs) or balanced (half of the inputs output $0$ and the other half output $1$). The DJ algorithm solves this problem by using a single call to the standard oracle implementing the function $f$ (with probability $1$). It works in the following way (for a single bit of input):

\begin{enumerate}
\item The player prepares two qubits in the $\ket{0}\ket{1}$.
\item It applies a Hadamard gate to the two qubits.
\item It applies $U_f$ with the second qubit receiving the output.
\item It applies a Hadamard to the first qubit.
\item It measures the first qubit in the computational basis and outputs the result.
\end{enumerate}

This is represented as the following circuit:

\[
\Qcircuit @C=1em @R=.7em @!R{
\lstick{\ket{0}} & \gate{\Ha} & \multigate{1}{U_f} & \gate{\Ha} & \meter \\
\lstick{\ket{1}} & \gate{\Ha} & \ghost{U_f} & \qw & \qw
}
\]

A simple calculation gives that the state right before the application of the last Hadamard on the first qubit is (with $b_i = f(i)$ for $i \in \{0, 1\}$ in the case of DJ for one input qubit):

\[
\frac{1}{\sqrt{2}}\bigl(\ket{0} + (-1)^{b_0 \oplus b_1}\ket{1}\bigr)\otimes\ket{-}
\]

%% file: tex_files/A1_OTP.tex
\subsection{Quantum Machines and Complexity Classes}
\label{subsec:qppt}

We first give the formal definition of the complexity class Bounded-Error Quantum Polynomial Time \cite{Yao93} ($\BQP$). This is considered to be the class of languages efficiently accepted by quantum machines (ie.\ it is the quantum equivalent of $\BPP$), or equivalently the class of decision problems efficiently solvable by quantum computers. By extension we use it to describe any efficient quantum computation, regardless of whether it solves a problem in the complexity-theoretic sense.

\begin{definition}[Languages in $\BQP$]
\label{bqp_m}
We say that a language $L$ is in $\BQP$ if and only if there exists a 
polynomial-time uniform family of quantum circuits \footnote{There exists a 
polynomial-time deterministic Turing machine taking as input $1^n$ for $n \in 
\mathbb{N}$ and outputting a classical description of $Q_n$.} $\qty{Q_n \mid n 
\in \mathbb{N}}$ such that:
\begin{itemize}
\item $Q_n$ takes $n$ qubits as input and outputs a single classical bit.
\item For all $x \in L$, $\mathbb{P}\Bigl[b = 1 \mid b \leftarrow Q_{\#x}(x) \Bigr] \geq \frac{2}{3}$.
\item For all $x \notin L$, $\mathbb{P}\Bigl[b = 0 \mid b \leftarrow Q_{\#x}(x) \Bigr] \geq \frac{2}{3}$.
\end{itemize}
\end{definition}

We now give the formal definition of a Quantum-Strong Probabilistic Polynomial-Time Machine ($\QPPT$), taken from 
\cite{Unr10}.

\begin{definition}[Quantum-Strong PPT Machine]
\label{qppt_m}
A classical machine $M$ is said to be \emph{Quantum-Strong Probabilistic
Polynomial-Time} (or $\QPPT$) if there exists a quantum polynomial-time (or
$\BQP$) machine $\tilde{M}$ such that for any classical network $N$ (as
defined in the Universal Composability Framework), $N \cup M$ and $N \cup
\tilde{M}$ are perfectly indistinguishable. For a given $\BQP$ machine $M$, we
define the \emph{$\QPPT$-reduced machine} $M_{\QPPT}$ as the machine applying
each round the same unitaries as machine $M$ to its internal registers but
always measuring its classical communication register in the computational
basis after the application of said unitary.
\end{definition}

\subsection{Proof of Lemma \ref{ind-to-sec}}
\label{subsec:proof_ind}

We now give the proof of Lemma~\ref{ind-to-sec} stated on 
page~\pageref{ind-to-sec}.

\setcounter{tempresult}{\value{lemma}}
\setcounter{lemma}{\value{compteur:ind-to-sec}-1}
\begin{lemma}[Input-Indistinguishability to Security]
Let $f \in \mathfrak{F}$ be the function to be computed by protocol $\Pi$
between parties $\Pone$ and $\Ptwo$, where $\mathfrak{F}$ is the set of
functions taking as input $(x, y) \in \{0, 1\}^{n_Y} \times \{0, 1\}^{n_X}$
and outputting $z \in \{0, 1\}^{n_Z}$ to $\Ptwo$ (and no output to $\Pone$). If
the protocol is input-indistinguishable for adversarial $\Pone^*$ in class
$\mathfrak{X}$ (Definition \ref{inp-ind}) then it is secure against
adversarial $\Pone^*$ in class $\mathfrak{X}$ (Definition \ref{sec-def}) with
identical bounds.
\end{lemma}
\setcounter{lemma}{\value{tempresult}}

\begin{proof}[Input-Indistinguishability to Security (Lemma \ref{ind-to-sec})]

If we suppose that the protocol is input-indistinguishable for a given class of Adversaries $\mathfrak{X}$, then no computationally-bounded quantum Distinguisher (represented as a an efficient quantum machine acting on the state returned by the Adversary) can distinguish between an execution with inputs $y_1$ and $y_2$. The Simulator against an Adversary in class $\mathfrak{X}$ then simply runs the protocol honestly with a random input $\tilde{y}$ (it does not need to call the Ideal Functionality as the adversarial player has no output). Therefore:

\[
\Bigl\lvert \mathbb{P}\Bigl[b = 0 \mid b \leftarrow \mathcal{D}\Bigl(v_{\mathcal{A}}(\Ptwo(y), \rho_{\mathcal{A}})\Bigr)\Bigr] -  \mathbb{P}\Bigl[b = 0 \mid b \leftarrow \mathcal{D}\Bigl(v_{\mathcal{A}}(S_{\Pone^*}(\tilde{y}), \rho_{\mathcal{A}})\Bigl)\Bigr] \Bigr\rvert \leq \epsilon(\eta)
\]

Since this is the case for any efficient distinguisher, it also means that the probability that the Adversary outputs a given bit as the guess for the real or ideal execution is the same up to $\epsilon$ in both cases. Therefore the protocol is secure.

\qed
\end{proof}

\subsection{Ideal Functionalities}
\label{app:IF}

We present here the Ideal Functionalities used throughout the paper, starting with the Two-Party Computation Ideal Functionality \ref{IdealF2PC} which Yao's Protocol implements.

\begin{idealf}[!htp]
\caption{Two-Party Secure Function Evaluation.}
\label{IdealF2PC}
\begin{itemize}

\item \textbf{Public information:} Binary function $f : \{0, 1\}^{n_X} \times 
\{0, 1\}^{n_Y} \longrightarrow \{0, 1\}^{n_Z}$ to be computed (where $n_X$, 
respectively $n_Y$, is the size of the input of $\Pone$, respectively 
$\Ptwo$, and $n_Z$ is the size of the output).

\item \textbf{Inputs:} $\Pone$ has classical input $x \in \{0, 1\}^{n_X}$ and $\Ptwo$ has classical input $y \in \{0, 1\}^{n_Y}$.

\item \textbf{Computation by the trusted party:}
\begin{enumerate}
\item If the trusted party receives an input which is inconsistent with the required format (different input size) or $\Abort$, it sends $\Abort$ to both parties. Otherwise, let $\tilde{\rho}_{in}$ be the input state it received from $\Pone$ and $\Ptwo$.
\item The trusted party measures the parts of $\tilde{\rho}_{in}$ in registers $\Regone$ and $\Regtwo$ in the computational basis, let $(\tilde{x}, \tilde{y})$ be the outcomes of the measurement.
\item The trusted party computes $\tilde{z} = f(\tilde{x}, \tilde{y})$ and sends $(\tilde{x}, \tilde{z})$ to $\Pone$ and $\tilde{y}$ to $\Ptwo$.
\end{enumerate} 

\end{itemize}
\end{idealf}

In the classical case, it is argued in 
\cite{LP04proof} that it suffices without loss of generality to describe the 
ideal functionality for functions where only one party receives an output, in this 
case $\Pone$, via the following transformation: any function inputs $(x, y)$ 
and two outputs $(w, z)$ to two parties can be transformed into a function 
taking as input $((x, p, a, b), y)$ and outputting to a single party $(w, 
\alpha:= z \oplus p,
\beta := a \odot \alpha \oplus b)$, where $\oplus$ and $\odot$ are the 
addition and 
multiplication operations in a well-chosen finite field ($p$ serves as a 
perfect One-Time-Pad of the output $z$ and $\beta$ serves as a perfect 
One-Time Message Authentication Code of $\alpha$). 
As shown in Section \ref{subsec:path} this is not so clear in our model. 

Then the 1-out-of-2 String Oblivious 
Transfer (Ideal Functionality \ref{IdealOT}), in which one 
party ($\Pone$ in our case) has two strings $(k_0, k_1)$ and the other 
($\Ptwo$) has a bit $b \in \{0, 1\}$. The output of $\Ptwo$ is the string 
$k_b$ (with no knowledge about $k_{\bar{b}}$), while on the other hand $\Pone$ 
has no output and no knowledge about choice-bit $b$. This Ideal Functionality is used by Yao's Protocol to implement the previous one.

\begin{idealf}
\caption{1-out-of-2 String OT.}
\label{IdealOT}
\begin{itemize}

\item \textbf{Inputs:} $\Pone$ has as input $(k_0, k_1)$ and $\Ptwo$ has as input $b \in \{0, 1\}$.

\item \textbf{Computation by the trusted party:}
\begin{enumerate} 
\item If the trusted party receives $\Abort$ or an incorrectly formatted input from either party, it sends $\Abort$ to both parties and halts.
\item Otherwise, let $(\widehat{k_0}, \widehat{k_1})$ and $\hat{b}$ be the inputs received. The Ideal Functionality sends $\widehat{k_{\hat{b}}}$ to $\Ptwo$ and halts.
\end{enumerate} 
\end{itemize}
\end{idealf}

\subsection{Superposition-Resistance of the Classical One-Time Pad}
\label{subsec:sec_otp}

As a cryptographic ``Hello World", we prove that the classical One-Time-Pad remains secure in our new model. The OTP (Protocol \ref{otp}) uses a Key Distribution (Ideal Functionality \ref{k-dist}, see \cite{Por17a}) for two parties to emulates a Confidential Channel (Ideal Functionality \ref{Sec-chan}), which assures that only the length of the message is leaked to the Eavesdropper but does not guarantee that it was not tampered with (see also \cite{DFPR14} and \cite{CMT13}). The security of the classical OTP Protocol against $\QPPT$ Adversaries is proven in \cite{DFPR14}.

\begin{idealf}[!htp]
\caption{Confidential Channel.}
\label{Sec-chan}
\begin{itemize}

\item \textbf{Inputs:} The Sender has a message $m \in \{0, 1\}^n$. The Receiver has no input and the Eavesdropper has an auxiliary input $\rho_{aux}$.

\item \textbf{Computation by the trusted party:} The Ideal Functionality sends $n$ to the Eavesdropper. If it has not received anything from the Eavesdropper, the trusted party sends $m$ to the Receiver. Otherwise if it has received quantum state from $\rho_{aux}$ from the Eavesdropper over $n$ qubits, it measures the state in the computational basis. Let $\hat{m}$ be the result of the measurement, it then sends $\hat{m}$ it to the Receiver.
\end{itemize}
\end{idealf}

\begin{idealf}[!htp]
\caption{Key Distribution.}
\label{k-dist}
\begin{itemize}

\item \textbf{Inputs:} Parties $\Pone$ and $\Ptwo$ have as input the size $n$ of the key.

\item \textbf{Computation by the trusted party:} It samples uniformly at random $k \in \{0, 1\}^n$ and sends $k$ to $\Pone$ and $\Ptwo$.
\end{itemize}
\end{idealf}

\begin{algorithm}[!htp]
\caption{OTP Protocol.}
\label{otp}
\begin{algorithmic}[0]
\STATE \textbf{Input:}
\begin{itemize}
\item The Sender has a message $m \in \{0, 1\}^n$. 
\item The Receiver has as input the size of the message.
\item The Eavesdropper has an auxiliary input $\rho_{aux}$.
\end{itemize}
\STATE \textbf{The Protocol:}
\begin{enumerate}
\item The Sender and Receiver call the Key Distribution Ideal Functionality on input $n$ and receive a key $k$ of size $n$.
\item The Sender computes $y = m \oplus k$ (where $\oplus$ corresponds to an bit-wise XOR) and sends it to the Eavesdropper.
\item The Eavesdropper sends a message $\hat{y}$ to the Receiver.
\item The Receiver compute $\hat{m} = \hat{y} \oplus k$
\end{enumerate}
\end{algorithmic}
\end{algorithm}

We will now prove the security of the protocol against malicious $\BQP$ Receiver (with superposition access), as captured by the following Lemma \ref{sec-otp}.

\begin{lemma}[Security of One-Time-Pad against Adversaries with Superposition Access]
\label{sec-otp}
Protocol \ref{otp} is superposition-resistant against a malicious Eavesdropper with advantage $\epsilon = 0$ (ie.\ it satisfies Definition \ref{sec-def} against $\BQP$ Adversaries).
\end{lemma}

\begin{proof}

We start by defining the quantum equivalent of all operations in Protocol \ref{otp}. The initial message is represented as a quantum register containing $\ket{m}$. The call to the Key Distribution Ideal Functionality yields a quantum register for both parties containing a state $\ket{k}$ in the computational basis. The bit-wise $XOR$ is applied using $\CNOT$ gates where the key corresponds to the control. The definition of the $\CNOT$ gate implies that if the control is in a computational basis state, it remains unentangled from the rest of the state after application of the gate. The state is then sent to the Eavesdropper. It can perform any CPTP map on the state $\ket{y}\otimes \rho_{aux} \otimes \ket{0^{\otimes n}}$ and send the last register to the Receiver. The Receiver applies the XOR using $\CNOT$ gates with its key as control.

The Eavesdropper has no output in this protocol. As stated in Lemma \ref{ind-to-sec}, it would be sufficient to show that two executions with different inputs are indistinguishable. However we will now describe the Simulator for clarity. It receives the size of the message $n$ from the Confidential Channel Ideal Functionality. It chooses uniformly at random a value $\tilde{y} \in \{0, 1\}^n$ and sends $\ket{\tilde{y}}$ to the Eavesdropper. It receives in return a state $\rho$ on $n$ qubits and sends it to the Confidential Channel Ideal Functionality (which then measures the state in the computational basis). 

Before the message sent by the Adversary, the protocol is equivalent to its classical execution, so the Adversary has no additional advantage compared to the classical execution (which is perfectly secure). The only advantage possibly obtained by the Adversary compared to a fully classical one comes from the state that it sent to the Receiver (repsectively Simulator) and the application by the Receiver of an operation dependent on its secret key (respectively a measurement in the computational basis by the Ideal Functionality). It is a well known fact (No-Communication Theorem of quantum information) that the Adversary obtaining any bit of information with probability higher than $0$ via this method (using only local operation on the Receiver's side or by the Ideal Functionality) would violate the no-signalling principle \cite{GGRW88,ER89}, therefore the distinguishing advantage of the Adversary between the real and ideal executions is $0$, thereby concluding the proof.

\qed
\end{proof}

%% file: tex_files/04b_symmetric_encryption.tex
Protocol \ref{Yao} (given on page~\pageref{Yao}) uses a symmetric encryption 
scheme as a primitive. We give here the definitions of the properties required 
from such a scheme for the correctness and security of the protocol. Recall 
that such a scheme is defined as two classical efficiently computable 
deterministic functions $(\Enc, \Dec)$ over $\mathfrak{K}$, $\mathfrak{A}$, 
$\mathfrak{M}$ and $\mathfrak{C}$. For simplicity we will suppose that the 
key-generation algorithm simply samples the key uniformly at random from the 
set of valid keys.

\subsubsection{Security and Superposition-Compatibility of Symmetric 
Encryption.}

The encryption scheme must satisfy the following properties in order to behave well when called on superpositions of inputs (so that they may be represented as minimal oracles, shown by Lemma \ref{min_or}). The first of those is that the encryption and decryption functions are perfect inverses of each other.

\begin{definition}[Correctness]
\label{perf-corr}
An encryption scheme $(\Enc, \Dec)$ as defined above is said to be \emph{correct} if, for all keys $k \in \mathfrak{K}$ and auxiliary input $\mathrm{aux} \in \mathfrak{A}$, $\Dec_{k}(\mathrm{aux}, \cdot) \circ \Enc_{k}(\mathrm{aux}, \cdot) = \mathsf{Id}_{\mathfrak{M}}$ where $\circ$ is the composition of functions and $\mathsf{Id}_{\mathfrak{M}}$ is the identity function over set $\mathfrak{M}$.

\end{definition}

It is also necessary that the plaintexts and ciphertexts belong to the same set.

\begin{definition}[Format-Preserving Encryption]
\label{format-pres}
An encryption scheme $(\Enc, \Dec)$ defined as above is said to be \emph{format-preserving} if $\mathfrak{M} = \mathfrak{C}$.

\end{definition}

The encryption scheme is called \emph{in-place} if no additional memory is required for performing encryptions and decryptions.

\begin{definition}[In-Place Permutations]
\label{in-place}

A permutation $\sigma$ over set $X$ is said to be \emph{in-place} if it can be computed efficiently using exactly the memory registers storing $x$ using reversible operations.

\end{definition}

Finally, the encryption scheme is called \emph{non-mixing} if the registers containing the key and auxiliary value are not modified in a way that depends on anything other than themselves. In another sense, the only mixing that is allowed is when modifying the message register during the encryption and decryption process.

\begin{definition}[Non-Mixing Encryption Scheme]
\label{non-mix}
Let $(k', \mathrm{aux'}, c) = \Enc_k(\mathrm{aux}, m)$ be the contents of the three memory registers at the end of the in-place encryption algorithm under key $k$ (transformed into $k'$ by the end of the encryption), where $m$ is the message encrypted into ciphertext $c$, $\mathrm{aux}$ is the auxiliary value and $\mathrm{aux'}$ is the content of the auxiliary register at the end of the encryption. The encryption scheme is said to be non-mixing if there exists two \emph{in-place} permutations $e_K : \mathfrak{K} \rightarrow \mathfrak{K}$ and $e_A : \mathfrak{A} \rightarrow \mathfrak{A}$ such that $k' = e_K(k)$ and $\mathrm{aux'} = e_A(\mathrm{aux})$ (and similarly for the decryption algorithm with functions $d_K$ and $d_A$).

\end{definition}

The function corresponding to the key may represent the key-expansion phase which is often present in symmetric encryption schemes (in which case $e_K$ and $d_K$ are injective functions from $\mathfrak{K}$ to a larger space, but the same results apply), while the one linked to the auxiliary value may be the updating of an initialisation value used in a block cipher mode of operation.

These four previous definitions ensure that the encryption and decryption algorithms can be represented as unitaries when acting on quantum systems without the use of ancillae (which is usual way of transforming a classical function into a quantum operation).

\begin{lemma}[Sufficient Conditions for Minimal Oracle Representation]
\label{min_or}
Let $(\Enc, \Dec)$ be a correct, format-preserving, in-place and non-mixing encryption scheme defined as above (satisfying Definitions \ref{perf-corr} through \ref{non-mix}). Then is has a Minimal Oracle Representation according to Definition \ref{def_min_or}. Furthermore for all superpositions $\ket{\phi} = \sum_{m \in \mathfrak{M}}\alpha_m \ket{m}$ (where $\ket{\tilde{\phi}}$ is the same superposition over encrypted values):

\[
 \mathsf{M}_{\Enc}\ket{k}\ket{\mathrm{aux}}\ket{\phi} = \ket{e_K(k)}\ket{e_A(\mathrm{aux})}\ket{\tilde{\phi}}
\]

\end{lemma}

\begin{proof}

The encryption scheme is correct and format preserving, which implies that, for every key $k \in \mathfrak{K}$ and any value $\mathrm{aux} \in \mathfrak{A}$, the functions $\Enc_{k}$ and $\Dec_{k}$ are permutations of $\mathfrak{M}$. In the case of a non-mixing encryption scheme, the functions $e_K$ and $e_A$ are also invertible and so the overall scheme is a permutation over $\mathfrak{K} \times \mathfrak{A} \times \mathfrak{M}$. Any such classical permutations can be represented as minimal oracles: given a permutation $\sigma$ over a set $X$ it is always possible, although costly, to compute $\sigma(x)$ for all $x \in X$ and define the minimal oracle $\mathsf{M}_{\sigma}$ by its matrix elements ($\mathsf{M}_{\sigma}[x][\sigma(x)] = 1$ and $0$ everywhere else).

The efficiency of the scheme lies in the fact that all these permutations are in-place, meaning that they can each be computed without using additional memory and using only reversible operations even in the classical case. The classical reversible operations can easily be implemented using unitaries (mainly $\X$, $\CNOT$ and Toffoli gates) and so the classical efficiency directly translates to the quantum case.

Finally, since the functions $e_K$ and $e_A$ do not depend on the message being encrypted (or decrypted in the case of $d_K$ and $d_A$), they remain unentangled from the message register if they were in a basis state prior to the application of the minimal oracle.

\qed
\end{proof}

The last definition below presents the security of the encryption scheme. In the following, the key-space, auxiliary-space and message-space are fixed to $\mathfrak{K} = \{0, 1\}^{n_K}$, $\mathfrak{A} = \{0, 1\}^{n_A}$ and $\mathfrak{M} = \{0, 1\}^{n_M}$ for some $n_K(\eta)$, $n_A(\eta)$ and $n_M(\eta) > n_K(\eta)$ (mainly for simplicity of presentation, the ideas transpose to other sets). We define the quantum security of such a symmetric encryption scheme by imposing that sampling the key and giving a black-box access to an encryption quantum oracle is indistinguishable for a quantum Adversary from giving it superposition access to a random permutation. This is simply a quantum game-based version of the definition for pseudo-random permutations \cite{Gol04prp}.

\begin{definition}[Real-or-Permutation Security of Symmetric Encryption]
\label{sec_sym}
Let $(\Enc, \Dec)$ be a symmetric encryption scheme with Minimal Oracle Representation. Let $\mathcal{S}_{n_M}$ be the set of permutations over $\{0, 1\}^{n_M}$. Consider the following game $\Gamma$ between a Challenger and the Adversary:

\begin{enumerate}
\item The Challenger chooses uniformly at random a bit $b \in \{0, 1\}$ and:
\begin{itemize}
\item If $b = 0$, it samples a key $k \in \{0, 1\}^{n_K}$ uniformly at random, and sets the oracle $\mathcal{O}$ by defining it over the computational basis states $\ket{\mathrm{aux}}\ket{m}$ for $m \in \{0, 1\}^{n_M}$ and $\mathrm{aux} \in \{0, 1\}^{n_A}$ as $\mathcal{O}\ket{\mathrm{aux}}\ket{m} = U_{\Enc}\ket{k}\ket{\mathrm{aux}}\ket{m} = \ket{k}\ket{e_A(\mathrm{aux})}\ket{\Enc_k(\mathrm{aux}, m)}$ (the oracle first applies the minimal encryption oracle $\mathsf{M}_{\Enc}$ and then the inverse of $d_K$ to the register containing the key).
\item If $b = 1$, it samples a permutation over $\{0, 1\}^{n_M}$ uniformly at random $\sigma \in \mathcal{S}_{n_M}$ and sets the oracle $\mathcal{O}$ as $\mathcal{O}\ket{\mathrm{aux}}\ket{m} = U_{\sigma, e_A}\ket{\mathrm{aux}}\ket{m} = \ket{e_A(\mathrm{aux})}\ket{\sigma(m)}$. 
\end{itemize}
\item For $i \leq q$ with $q = \mathrm{poly}(\eta)$, the Adversary sends a state $\rho_i$ of its choice (composed of $n_M$ qubits) to the Challenger. 
The Challenger responds by sampling an auxiliary value at random $\mathrm{aux}_i \in \{0, 1\}^{n_A}$, applying the oracle 
to the state $\ket{\mathrm{aux}_i} \otimes \rho_i$ and sending the result back to the Adversary along with the modified auxiliary value (notice that the oracle has no effect on the key if there is one and so it remains unentangled from the Adversary's system).
\item The Adversary outputs a bit $\tilde{b}$ and stops.
\end{enumerate}

A symmetric encryption scheme is said to be secure against quantum Adversaries if there exists $\epsilon(\eta)$ negligible in $\eta$ such that, for any $BQP$-Adversary $\mathcal{A}$ with initial auxiliary state $\rho_{\mathit{aux}}$:

\[
\mathit{Adv}_{\Gamma}(\mathcal{A}) := \abs{\frac{1}{2} - \mathbb{P}\Bigl[b = \tilde{b} \mid \tilde{b} \leftarrow \mathcal{A}(\rho_{\mathit{aux}}, \Gamma)\Bigr]} \leq \epsilon(\eta)
\]

\end{definition}

We show below how to instantiate an encryption scheme satisfying these definitions.


\subsubsection{Instantiating the Symmetric Encryption Scheme.}
\label{app:enc_inst}

In a sense, the perfect (but inefficient) symmetric encryption is given by associating each key $k \in [n_M!]$ to a different permutation from $\mathcal{S}_{n_M}$ in a canonical way (sampling the key is then equivalent to sampling the permutation). The encryption scheme that is used in the protocol may even be considered to be exactly this perfect encryption scheme since the superposition attack does not use the specifics of the underlying encryption scheme, or even supposes a negligible advantage in breaking the encryption scheme (it simply requires it to have a Minimal Oracle Representation).

Note however that a large number of symmetric encryption schemes can be made to fit these conditions. For example, the most widely used block-cipher AES \cite{DR02} operates by performing during a certain number of rounds $N$ the following process (the message consists of one block of $128$ bits, presented a four-by-four square matrix with $8$ bits in each cell):

\begin{enumerate}
\item It applies a round key (different for each round) by XORing it into the message-block.
\item It applies to each cell a fixed permutation $S : \mathbb{Z}_8 \leftarrow \mathbb{Z}_8$ over $8$ bits.
\item Each row $i$ is shifted cyclically by $i$ places to the left (the indices start at $0$).
\item Each column is multiplied by a fixed invertible matrix.
\end{enumerate}

This is clearly a permutation of the message block which is also in-place and non-mixing if implemented without optimisations (note that there is no auxiliary value apart from the invertible matrix and the key remains unchanged during the rounds). The security of AES is well studied classically and cryptanalysis has been performed recently in the quantum setting in \cite{BNPS19} (albeit not against Adversaries with superposition access).

The only place where the AES cipher is not in-place is during the 
key-derivation phase, during which the round keys are generated. The simple 
way to fix this is to make the register containing the original key large 
enough to contain the expanded key as well and initialise the additional 
qubits in the $\ket{0}$ state. The operation producing the larger key from the 
initial key then corresponds to the functions $e_K$ and $d_K$ from 
Definition~\ref{non-mix} (as mentioned in the remark below the definition, 
these can then be injective only without changing the result).

If the message to be encrypted is longer than a block, it can easily be extended by using the CBC operation mode, which is secure under the assumption that the underlying block-cipher is a quantum-secure pseudo-random permutation (based on the security analysis of \cite{ATTU16}). In this mode, the Initialisation Value (or $IV$) is a uniformly random string of the same size as the blocks upon which the block-cipher operates (it corresponds to the auxiliary value discussed previously). The encryption of the CBC mode operates by applying the function $c_i = \Enc_k(m_i \oplus c_{i - 1})$ for all $i$, where $m_i$ is the message block of index $i$ and $c_i$ is the corresponding ciphertext (with the convention that $c_0 = IV$), and $\Enc$ is the encryption algorithm of the underlying block-cipher. Conversely the decryption is given by $m_i = c_{i - 1} \oplus \Dec_k(c_i)$ with the same conventions. This is also clearly in-place and non-mixing (the IV and key are never modified so $e_A = d_A = Id_A$ and $e_K = d_K = Id_K$) as well as secure.

However, as mentioned above, the CBC mode of operation is only secure with superposition access under the assumption that the underlying block-cipher is also secure against Adversaries with superposition access. To our knowledge, no block-ciphers have been proven secure under this type of access (and some have been broken, such as the 3-round Feistel cipher in \cite{KM10} and the Even-Mansour cipher in \cite{KM12}, although both were patched in \cite{AR17} based on the Hidden Shift Problem). The CTR mode of operation on the other hand has been proven quantum-secure in \cite{ATTU16} even if the underlying block-cipher is only secure against quantum Adversaries with \emph{classical} access (which, as mentionned above, has been done in \cite{BNPS19} for AES). This mode also satisfies all the requirements stated in the definitions above with only a slight modification. The $IV$ in this mode is also initialised to a uniformly random value of the same size the blocks to be encrypted. The $IV$ is encrypted with the key and XORed to the first block, then the $IV$ in incremented and the same process is repeated for the subsequent blocks. This means that the IV needs to be copied (it can be done quantumly using a $\CNOT$ gate but it is then no longer in-place). In order to be in-place and non-mixing, the following procedure may be applied instead for each block (with $\mathcal{IV}$ being the register containing the $IV$ and $\mathcal{B}_i$ the register containing block at position $i$, with $\leftarrow$ being used for the assignment operator and $\mathsf{Inc}$ corresponding to the incrementation operator):
\\

\begin{tabular}{rccc}
1. & $\mathcal{IV}$ & $\leftarrow$ & $\Enc_k(\mathcal{IV})$\\
2. & $\mathcal{B}_i$ & $\leftarrow$ & $\mathcal{B}_i \oplus \mathcal{IV}$\\
3. & $\mathcal{IV}$ & $\leftarrow$ & $\Dec_k(\mathcal{IV})$\\
4. & $\mathcal{IV}$ & $\leftarrow$ & $\mathsf{Inc}(\mathcal{IV})$
\end{tabular}
\\

The overall result is that the value contained in the IV has only been updates by incrementing it as many times as there are blocks $n_b$ to be encrypted and therefore the function $e_A = d_A = \mathsf{Inc}^{\circ n_b}$ only acts on the auxiliary value $IV$.

Note that these methods work for messages whose length is a multiple of the block-size, but can be further extended by using a correctly chosen padding.

%% file: tex_files/04c_new_yao_protocol.tex
We first present in Protocol \ref{Yao} the formal version of our Modified Yao Protocol for a single bit of output.

\begin{algorithm}[ht!]
\caption{Modified Yao Protocol for One Output Bit.}
\label{Yao}
\begin{algorithmic}[0]
\STATE \textbf{Input:} The Garbler and Evaluator have inputs $x \in \{0, 1\}^{n_X}$ and $y \in \{0, 1\}^{n_Y}$ respectively.
\STATE \textbf{Output:} The Garbler has one bit of output, the Evaluator has no output.
\STATE \textbf{Public Information:} The function $f$ to be evaluated, the encryption scheme $(\Enc, \Dec)$ and the size of the padding $p$.
\STATE \textbf{The Protocol:}
\begin{enumerate}
\item The Garbler chooses uniformly at random the values $\qty{k_0^{G, 1}, k_1^{G, 1}, \ldots, k_0^{G, n_X}, k_1^{G, n_X}}$, $\qty{k_0^{E, 1}, k_1^{E, 1}, \ldots, k_0^{E, n_Y}, k_1^{E, n_Y}}$ from $\mathfrak{K}$ and $k^z \in \{0, 1\}$. It uses those values to compute the garbled table $\mathit{GT}_{f}^{(X, Y, Z)}$, with $X$ being the set of wires for the Garbler's input, $Y$ the set of wires for the evaluators input, and $Z$ the output wire.
\item The Garbler and Evaluator perform $n_Y$ interactions with the trusted third party performing the OT Ideal Functionality. In interaction $i$:
\begin{itemize}
\item The Garbler's inputs are the keys $(k_0^{E, i}, k_1^{E, i})$.
\item The Evaluator's input is $y_i$.
\item The Evaluator's output is the key $k_{y_i}^{E, i}$.
\end{itemize}
\item The Garbler sends the garbled table $\mathit{GT}_{f}^{(X, Y, Z)}$ and $2^{n_X + n_Y}$ copies of the keys corresponding to its input $\qty{k_{x_i}^{G, 1}}_{i \in [n_X]}$. It also sends the auxiliary values $\qty{\mathrm{aux}_i}_{i \in [n_X + n_Y]}$ that were used for the encryption of the garbled values.
\item For each entry in the garbled table:
\begin{enumerate}
\item The Evaluator uses the next ``fresh" copy of the keys supplied by the Garbler along with the keys that it received from the OT Ideal Functionality to decrypt the entry in the garbled table.
\item It checks that the last $p$ bits of the decrypted value are all equal to $0$. If so it returns the register containing the output value and the ones containing the Garbler's keys to the Garbler.
\item Otherwise it discards this ``used" copy of the keys and repeats the process with the next entry in the garbled table. If this was the last entry it outputs $\Abort$ and halts.
\end{enumerate}
\item If the Evaluator did no output $\Abort$, the Garbler applies the One-Time-Pad defined by the key associated with wire $z$ to decrypt the output: if $k^z = 1$, it flips the corresponding output bit, otherwise it does nothing. It then sets the bit in the output register as its output.
\end{enumerate}
\end{algorithmic}
\end{algorithm}


We then give the proofs of correctness (Theorem \ref{prot_corr}) and 
$\QPPT$-security (Theorem \ref{sec_yao}) of Protocol \ref{Yao} against 
adversarial Honest-but-Curious Sender and Receiver (as defined in our new model through 
Definition \ref{sec-def}). The modifications to the original protocol 
therefore do not impact security against classical or $\QPPT$ Adversaries.

\subsubsection{Proof of Correctness of Protocol \ref{Yao}}

\setcounter{tempresult}{\value{theorem}}
\setcounter{theorem}{\value{compteur:prot_corr}-1}
\begin{theorem}[Protocol Correctness]
Let $(\Enc, \Dec)$ be a symmetric encryption scheme with a Minimal Oracle Representation 
(Definition \ref{def_min_or}). Protocol \ref{Yao} is correct with probability exponentially close to $1$ in 
$\eta$ for $p = \mathit{poly}(\eta)$.
\end{theorem}
\setcounter{theorem}{\value{tempresult}}

\begin{proof}[Correctness of Yao's Protocol (Theorem \ref{prot_corr})]
We suppose here that both players are honest. Note that the protocol will only fail if one decryption which should have been incorrectly decrypted is instead decrypted as valid. The parameter $p$ must be chosen such that the probability of failure is negligible (in the security parameter in this instance). If at least one of the keys used in decrypting an entry in the garbled table does not correspond to the key used in encrypting it, the encryption and decryption procedure is equivalent to applying a random permutation on $r \parallel 0^p$ for uniformly random $r$ (up to negligible probability in $\eta$ that the encryption scheme is distinguishable from a random permutation). The probability that the resulting element also has $p$ bits equal to $0$ at the end is therefore $2^{-p}$.

For $p = \mathit{poly}(\eta)$, we show that the failure probability corresponding to one such event happening across any possible ``wrong" decryption is negligible in $\eta$. In fact, there are $2^{n_X + n_Y + 1}$ ciphertexts (counting both possibilities for $k^z$) and $2^{n_X + n_Y}$ possible input key combinations, all but one being wrong for each ciphertext. This results in $2^{n_X + n_Y + 1}(2^{n_X + n_Y} - 1) \approx 2^{2n_X + 2n_Y + 1}$ random values being potentially generated through incorrect decryption. The probability that none of these random values has the string $0^p$ as suffix (let $\mathsf{Good}$ be the associated event) is given by:

\[
\mathbb{P}[\mathsf{Good}] \approx \bigl(1 - 2^{-p}\bigr)^{2^{2n_X + 2n_Y + 1}} \approx 1 - 2^{-p} \cdot 2^{2n_X + 2n_Y + 1}
\]

The first approximation comes from the aforementioned negligible probability that the encryption scheme is not a random permutation while the second stems from $p \gg n_X + n_Y$. This probability should be negligibly close to $1$ in $\eta$, in which case setting $p = \eta + 2(n_X + n_Y)$ is sufficient.

\qed
\end{proof}

\subsubsection{Proof of $\QPPT$-Security of Protocol \ref{Yao}}

\setcounter{tempresult}{\value{theorem}}
\setcounter{theorem}{\value{compteur:sec_yao}-1}
\begin{theorem}[$\QPPT$-Security of Yao's Protocol]
Consider a hybrid execution where the Oblivious Transfer is handled by a 
classical trusted third party.
Let $(\Enc, \Dec)$ be a symmetric encryption scheme that is $\epsilon_{Sym}$-real-or-permutation-secure 
(Definition \ref{sec_sym}). Then Protocol 
\ref{Yao} is perfectly-secure against a $\QPPT$ adversarial Garbler (the 
Adversary's advantage is $0$) and $(2^{n_X + n_Y} - 1) 
\epsilon_{Sym}$-secure against $\QPPT$ adversarial Evaluator according to 
Definition \ref{sec-def}.
\end{theorem}
\setcounter{theorem}{\value{tempresult}}

\begin{proof}[$\QPPT$-Security of Yao's Protocol (Theorem \ref{sec_yao})]

In both cases (adversarial Garbler and Evaluator) we will construct a Simulator that runs the Adversary against the real protocol internally and show that the Adversary's advantage in distinguishing the real and ideal executions is negligible. First note that by the definition of the $\QPPT$ Adversary, there exists a $\BQP$ machine such that the two are indistinguishable in any network, meaning in particular in a classical network. This implies that the machine always sends and receives messages in the computational basis when in a quantum network.

\paragraph{Security against $\QPPT$ Garbler.} The Simulator works as follows:
\begin{enumerate}
\item During each OT, it performs the same interaction as an honest player would, but with a random value for the input $\tilde{y}_i$ of each OT.
\item The Adversary's machine then necessarily sends the Garbler's keys and the circuit in the computational basis.
\item This automatically fixes the value of the Adversary's input $\hat{x}$ (the Adversary being Honest-but-Curious, it has generated the keys correctly and sent the keys corresponding to its input). The Simulator can therefore measure the register containing the input of the Garbler to recover $\hat{x}$.
\item The Simulator then sends $\hat{x}$ to the Ideal Functionality computing the function $f$ and gets $f(\hat{x}, \hat{y})$ (for the actual value of the honest player's input $\hat{y}$).
\item The Simulator can compute the value $f(\hat{x}, \tilde{y})$ and decrypt the garbled table values to recover $f(\hat{x}, \tilde{y}) \oplus k^z$ using the keys that were giving to it through the OTs (for its ``fake" input $\tilde{y}$). It uses both values to recover $k^z$.
\item The Simulator then computes $f(\hat{x}, \hat{y}) \oplus k^z$ and sends this value to the Adversary.
\end{enumerate}

The only advantage of the Adversary in the real protocol compared to this ideal execution stems from its potential usage of the execution of the OT protocols for distinguishing the real and ideal world. This execution is ideal in the hybrid model and so the advantage of the Adversary is $0$. 

\paragraph{Security against $\QPPT$ Evaluator.} The messages sent to the adversarial Evaluator consist of $n_Y$ instances of OTs, $2^{n_X + n_Y}$ garbled table entries and the keys corresponding to the input of the honest player. The Simulator performs all these steps similarly to an honest Garbler but sends the keys corresponding to a randomly chosen input $\tilde{x}$. We can show through a series of games that this does not give any information to a computationally-bounded Evaluator (we show that the protocol is input-indistinguishable according to Definition \ref{inp-ind}, which as stated Lemma \ref{ind-to-sec} is equivalent since the Adversary has no output):

\begin{itemize}
\item \textbf{Game 0:} The Simulator performs Protocol \ref{Yao} with the Adversary, with random input $\tilde{x}$.
\item \textbf{Game 1:} In the execution of the OTs the Simulator replaces the values of the keys that are not chosen by the Adversary with random values (that were not used to compute any of the encryptions). The advantage in the real-world for the Adversary compared to this situation is $0$ since the execution of the OTs is perfectly-secure in the hybrid model. 
\item \textbf{Game 2:} The encryptions that use those (now random) keys can be replaced by random values, with a security cost of $\epsilon_{Sym}$ per replaced encryption (as the encryption can be considered to be random permutations without having access to the key). It is a double encryption, so for some values the Adversary may posses either the inner or the outer key. This means that it could invert one of the encryptions, but since it does not have the other this is meaningless.
\item \textbf{Game 3:} Finally, the key $k^z$ only appears in one encryption as a One-Time-Pad of the output of the computation (the others are now independent from it). It can therefore be replaced by an encryption of a random value, meaning that it is as well a random value (this is perfectly equivalent).
\end{itemize}

Finally, at the end of Game $3$, only the keys received through the OT remain and they are random values chosen independently from one another and from any input. The Adversary has no advantage in this scenario, meaning that the overall advantage is at most $(2^{n_X + n_Y} - 1) \epsilon_{Sym}$.

\qed
\end{proof}

The proof above shows that proving the security of some protocols does not require the Simulator to call the Ideal Functionality, in particular if the adversarial party does not have an output in the protocol. This is contrary to the usual simulation-based proofs, where the Simulator must extract the input of the Adversary to send it to the Ideal Functionality (for the sake of composition). However, the exact same proofs of security work in the Stand-Alone Framework of \cite{HSS15} if the Simulator does send the input value of the Adversary to the Ideal Functionality (any Adversary against a classical protocol in the Stand-Alone Framework is $\QPPT$ as well).

%% file: tex_files/A5_form_att.tex
In \ref{Qembed} we first describe how the actions performed during the 
protocol are transcribed into quantum operations, then give the formal description 
of the superposition attack in two parts. 
In \ref{subsec:supp_gen} we first describe the 
actions of the Adversary during the protocol, resulting in a state that contains a 
superposition of inputs and outputs, 
followed by how 
it is used by the Adversary to break the protocol by extracting one bit of information 
from the honest player's input. Finally, in \ref{subsec:comments} elaborates on a few 
comments about the attack that are sketched in the main text.

\subsubsection{Quantum Embedding of the Classical Protocol.}
\label{Qembed}

The inputs of each party are stored in one register each, as $\ket{x}$ and $\ket{y}$ respectively. For each key $k$ that is created as part of the protocol, a different quantum register is initialised with the state $\ket{k}$ (there are therefore $n_Y$ registers for the Evaluator's keys and $n_X2^{n_X + n_Y}$ for the Garbler's keys due to the copies being generated). Similarly, for each value $E_i$ of the garbled tables, a quantum register is initialised with the value $\ket{E_i}$ (there are $2^{n_X + n_Y}$ such registers). The auxiliary values are also all stored in separate quantum registers. All of these values are encoded in the computational basis.

The OT trusted party works as described in the Ideal Functionality \ref{IdealOT} in Appendix \ref{app:OTP}. The inputs and outputs are considered to be pure quantum states in the computational basis (no superposition attack is allowed to go through the OT). Sending messages in the other parts of the protocol is modelled as taking place over perfect quantum channels (no noise is present on the channel and superpositions are allowed to pass undisturbed). A decryption of ciphertext $c$ using a key $k$ and auxiliary value $\mathrm{aux}$ is modelled using the Minimal Oracle Representation from Definition \ref{def_min_or} as $\mathsf{M}_{\Dec}\ket{k}\ket{\mathrm{aux}}\ket{c} = \ket{d_K(k)}\ket{d_A(\mathrm{aux})}\ket{\Dec_k(\mathrm{aux}, c)}$ on the states of the computational basis.

Checking whether the final $p$ bits are equal to $0$ corresponds to performing 
a measurement $\mathcal{M}_C$ on the corresponding register $\mathcal{P}$ in 
the basis $\qty{\dyad{0^p}, \Id_{\mathcal{P}} - \dyad{0^p}}$. If the 
measurement fails, the Evaluator applies the inverses of $d_K$ and $d_A$ to 
the registers containing respectively its keys and the auxiliary values so 
that they may be reused in the next decryption. Finally, the correction 
applied at the end which depends on the choice of key for wire $Z$ is modelled 
as classically controlled Pauli operators $\X^{k^z}$ (this corresponds to the 
quantum application of a classical One-Time-Pad and the value $k^z$ can be 
seen as internal classical values of the Garbler for simplicity).

For simplicity of notations, let $k_y^E := k_{y_1}^{E, 1} \parallel \ldots 
\parallel k_{y_{n_Y}}^{E, n_Y}$ for $y \in \{0, 1\}^{n_Y}$ (and similarly for 
$x \in \{0, 1\}^{n_X}$). Also, let $\widetilde{\Enc}$ be the sequential 
encryption by all keys corresponding to strings $x$ and $y$, using the set of 
auxiliary values $\widetilde{\mathrm{aux}} := \mathrm{aux}_1 \parallel \ldots 
\parallel \mathrm{aux}_{n_X + n_Y}$. Then $E_{x, y}^{k^z} = 
\widetilde{\Enc}_{k_x^G, k_y^E}(\widetilde{\mathrm{aux}}, f(x, y) \oplus c 
\parallel 0^p)$. Finally, $\widetilde{d_K}$ is the function applying $d_K$ to 
each key, and similarly for $\tilde{d}_A$.

\subsubsection{Formal Presentation of the Superposition Attack.}
\label{subsec:supp_gen}
\label{subsec:full_att}

We give the formal presentation of the Superposition Generation Procedure in 
Attack~\ref{Yao_att} 
and then of the full attack on the Modified Yao Protocol in 
Attack~\ref{Yao_full_att}. We show that this Adversary is Extended 
Honest-but-Curious (Definition~\ref{ext-hbc}) and therefore that its 
$\QPPT$ equivalent would not break Theorem \ref{sec_yao}.

\begin{attack}[ht!]
\caption{Superposition Generation Procedure on the Modified Yao Protocol.}
\label{Yao_att}
\begin{algorithmic}[0]
\STATE \textbf{Inputs:}
\begin{itemize}
\item The (adversarial) Garbler has as input the quantum state $\phi_{\mathit{inp}, G} = \ket{\widehat{x_0}} \otimes \ket{\widehat{x_1}}$, with $\widehat{x_0}, \widehat{x_1} \in \{0, 1\}^{n_X}$ received from Environment $\mathcal{Z}$ (this describes classically the superposition of inputs that it should use).
\item The (honest) Evaluator has as input $\hat{y} \in \{0, 1\}^{n_Y}$, received from Environment $\mathcal{Z}$.
\end{itemize}
\STATE \textbf{The Attack:}
\begin{enumerate}
\item The Garbler chooses uniformly at random the values $\qty{k_0^{G, i}, k_1^{G, i}}_{i \in [n_X]}$ and $\qty{k_0^{E, i}, k_1^{E, i}}_{i \in [n_Y]}$ and computes the \emph{initial} garbled tables $\mathit{GT}_{f}^{(X, Y, Z), 0} = \qty{E_{x, y}^0}_{x, y}$ and $\mathit{GT}_{f}^{(X, Y, Z), 1} = \qty{E_{x, y}^1}_{x, y}$ (where the index $0$ corresponds to $k^z = 0$ and similarly for $1$, or equivalently that the value encrypted is $f(x, y)$ in the first case and $f(x, y) \oplus 1$ in the second). This computation is the same as in the honest protocol (but done for both values of $k^z$). Note that there is no need to permute the values as they will be sent in superposition anyway.
\item The Garbler and Evaluator perform $n_Y$ interactions with the trusted third party performing the OT Ideal Functionality. At the end of all interactions, the Evaluator has a quantum register initialised in the state $\smashoperator{\bigotimes\limits_{i \in [n_Y]}}\ket{k_{\hat{y}_i}^{E, i}} = \ket{k_{\hat{y}}^E}$.
\item The Garbler sends the auxiliary values as it would in the original protocol. The corresponding state is $\ket{\widetilde{\mathrm{aux}}} = \smashoperator{\bigotimes\limits_{i = 1}^{n_X + n_Y}}\ket{\mathrm{aux}_i}$. For each key $k_x^G$ that it would send to the Evaluator, the Garbler instead sends a uniform superposition $\frac{1}{\sqrt{2}}\bigl(\ket{k_{\widehat{x_0}}^G} +\ket{k_{\widehat{x_1}}^G}\bigr)$. For each entry in the garbled table that it would send, it instead sends the following superposition over all garbled values:

\[
\ket{\mathit{GT}} = \frac{1}{\sqrt{2^{n_X + n_Y + 1}}}\smashoperator{\sum\limits_{x, y, k^z}} (-1)^{k^z}\ket{E_{x, y}^{k^z}}
\]

\item For each entry in the garbled table, the Evaluator proceeds as it would in the protocol, decrypting the ciphertexts sequentially, performing a measurement on register $\mathcal{P}$ in the basis $\qty{\dyad{0^p}, \Id_{\mathcal{P}} - \dyad{0^p}}$ and returning the corresponding output and the register containing the Garbler's keys if successful.
\item If the Evaluator is successful and returns a state after one of its measurements, the Garbler applies the following clean-up procedure:
\begin{enumerate}
\item For each register containing one of its keys, it applied the inverse of $d_K$.
\item For each index $i$ such that $\widehat{x_0}^i \neq \widehat{x_1}^i$, if there is an index $j$ such that $k_0^{G, i, j} \neq k_1^{G, i, j}$ and $k_0^{G, i, j} = 1$ it applies an $\X$ Pauli operation on the qubit containing this bit of the key.
\item The first register then contains a superposition of logical encodings of the inputs $\widehat{x_0}^{L'}$ and $\widehat{x_1}^{L'}$. The register containing the output is unchanged.
\end{enumerate}
\item It then sets these registers (called \emph{attack registers}) as its output, along with a register containing $\ket{\widehat{x_0}} \otimes \ket{\widehat{x_1}} \otimes \ket{\mathbf{L'}}$, with $\mathbf{L'}$ being a list of integers corresponding to the size of a given logical repetition encoding of the inputs (see the proof of Theorem \ref{correct_att}).
\end{enumerate}
\end{algorithmic}
\end{attack}



\begin{attack}[ht!]
\caption{Full Superposition Attack on the Modified Yao Protocol.}
\label{Yao_full_att}
\begin{algorithmic}[0]
\STATE \textbf{The Attack:}
\begin{enumerate}
\item The Environment $\mathcal{Z}$ generates values $(\widehat{x_0}, \widehat{x_1}, \hat{y})$ and sends $\ket{\widehat{x_0}} \otimes \ket{\widehat{x_1}}$ to the Adversary. The values $(\widehat{x_0}, \widehat{x_1})$ are non-trivial in the sense that they do not uniquely determine the value of the output.
\item The Adversary applies the Superposition Generation Procedure described in Attack \ref{Yao_att}, using a superposition of keys for $\widehat{x_0}$ and $\widehat{x_1}$. If the Evaluator was not successful, the Adversary samples and outputs a bit $b$ (equal to $0$ with probability $p_{\mathit{Guess}}$, which corresponds to the optimal guessing probability and whose value is defined in the proof of Theorem \ref{sup_att_analysis}) and halts. 
\item Otherwise, the Adversary applies the following clean-up procedure on the output state of the Superposition Generation Procedure, similar to the one described in Attack \ref{Yao_att} (recall that the first register then contains a logical encoding of the inputs $\widehat{x_0}^{L'}$ and $\widehat{x_1}^{L'}$, obtained after the first clean-up procedure described in Attack \ref{Yao_att}):
\begin{enumerate}
\item If there is an index $j$ such that $\widehat{x_0}^j \neq \widehat{x_1}^j$ and $\widehat{x_0}^j = 1$, it applies a Pauli $\X$ operation on the qubits corresponding to the logical encoding of bit $j$ (each bit is encoded with a repetition code of varying length given by the list $\mathbf{L'}$).
\item The qubits corresponding to a value $j$ such that $\widehat{x_0}^j = \widehat{x_1}^j$ are unentangled from the rest of the state and so can be discarded.
\end{enumerate}
\item The result of the previous step is that the first register now contains a superposition of a logical encoding $0^L$ and $1^L$ (for another logical encoding $L$). The Adversary then applies a logical Hadamard gate $\Ha^L$ on this register.
\item The Adversary measures the first qubit in the computational basis and outputs the result.
\end{enumerate}
\end{algorithmic}
\end{attack}

\paragraph{Adversarial Behaviour.} Note that the Original Yao Protocol is 
secure against Honest-but-Curious Adversaries. The equivalent in terms of 
superposition attacks is to send exactly the same messages but computed over 
an arbitrary superposition of the randomness used by the Adversary (be it the 
inputs of other random values). That is to say that, if the honest party would 
have measured the state sent by the Adversary, it would recover perfectly 
honest classical messages. The Adversary described in Attack 
\ref{Yao_full_att} is not strictly Honest-but-Curious but we prove below that 
its classical-message equivalent does not break the security of the Modified Yao Protocol.

\setcounter{tempresult}{\value{lemma}}
\setcounter{lemma}{\value{compteur:lem:adv_beh}-1}
\begin{lemma}[Adversarial Behaviour Analysis]
The \emph{$\QPPT$-reduced machine} (Definition \ref{qppt_m}) corresponding to
the $\BQP$ Adversary described in Attack \ref{Yao_full_att} is \emph{Extended
Honest-but-Curious Adversary} (Definition \ref{ext-hbc}).
\end{lemma}
\setcounter{lemma}{\value{tempresult}}

\begin{proof}[Adversarial Behaviour Analysis (Lemma \ref{lem:adv_beh})]
Attack \ref{Yao_full_att} is not strictly 
Honest-but-Curious since a player that measures honestly and tries to decrypt 
after can also fail with probability $1 - e^{-1}$ if it never gets the correct 
ciphertext in the table after measuring. 
	The \emph{$\QPPT$-reduced machine} (Definition \ref{qppt_m}) corresponding 
to this $\BQP$ Adversary works as follows:
\begin{enumerate}
\item It generates all values for the garbled table (for both values of $k^z$).
\item For each garbled table entry that it is supposed to send, it instead chooses uniformly at random one of the generated values (with 
replacement) and a key for either $\widehat{x_0}$ or $\widehat{x_1}$ and sends them (it does not store in memory which values have been sent).
\item It then waits to see if the honest player has been able to decrypt one of the values or not.
\item If it has, then it receives (as classical messages) the key that was used to decrypt (either for $\widehat{x_0}$ or $\widehat{x_1}$) and the decrypted value.
\end{enumerate}

This Adversary is precisely an 
Extended Honest-but-Curious Adversary according to Definition \ref{ext-hbc} as 
the Simulator presented in the security proof of Theorem \ref{sec_yao} works 
as well for this Adversary, with the difference that with a probability of $1 
- e^{-1}$ it cannot recover the value of $k^z$ if it is unable to decrypt (but 
then this is also the case when interacting with an honest party) and so must 
abort. Since the Adversary does not store which values have been sent it does 
not know whether this value has been decrypted from the keys from the honest 
player or the Simulator (using a random input). On the other hand this action by 
the Simulator is necessary to simulate the probability that none of the keys decrypt 
correctly the garbled values (this happens with the same probability in the simulated 
and real executions)

\qed
\end{proof}

The core reason why the Honest-but-Curious Simulator works is that the Adversary's 
internal register is never entangled with the states that are sent to the 
honest party: much more efficient attacks exist in that case, for example the 
Adversary can recover the full input of the Evaluator if it keeps a register 
containing the index of the garbled table value, which collapses along with 
the output register when it is measured by the honest player while checking 
the padding, therefore revealing the input of the Evaluator. However this 
Adversary is not simulatable after being translated into a $\QPPT$ machine 
(and therefore this attack does not show a separation between the two scenarii 
as it would be similar to subjecting the protocol to a Malicious classical 
Adversary, that can trivially recover the honest player's input).

\subsubsection{Comments on the Superposition Attack.}
\label{subsec:comments}

We now give a few comments on the superposition generated by the Adversary, as 
described above and in Section~\ref{state_gen}. We first show how to generalize this 
generation for an (almost) arbitrary superposition. Then, we justify some of the 
choices made in the design of the variant of Yao's protocol 
(Section~\ref{prot_prez}).

\paragraph{Generalisation for Almost Separable Superpositions.} For binary function $f: \{0, 1\}^{n_X} \times \{0, 1\}^{n_Y} \rightarrow \{0, 1\}$ and $\hat{y}$, let $U_f^{\hat{y}}$ be the unitary defined through its action of computational basis states by $U_f^{\hat{y}}\ket{x}\ket{k^z} = \ket{x}\ket{f(x, y) \oplus k^z}$. The above procedure allows the Adversary to generate $U_f^{\hat{y}}\ket{\psi}\ket{\phi}$ for any states $\ket{\psi}$ (over $n_X$ qubits) and $\ket{\phi}$ (over one qubit) whose classical descriptions $\psi$ and $\phi$ are efficient (notice that the state $\ket{\psi}\ket{\phi}$ must be separable). The description of state $\psi$ is used to generate the superposition of keys (if an input appears in the superposition $\psi$, then the key corresponding to it should appear in the superposition of keys with the same amplitude) while $\phi$ is used when generating the superposition over garbled table entries (if $\ket{\phi} = \alpha\ket{0} + \beta\ket{1}$, the corresponding superposition over garbled values is $\ket{GT_{\alpha, \beta}} =  \smashoperator{\sum\limits_{x, y}}\alpha\ket{E_{x, y}^0} + \beta\ket{E_{x, y}^1}$). The same results and bounds are applicable (with similar corresponding proofs).

\paragraph{Justifying the Differences in the Protocol Variant.} We can now 
more easily explain the choices straying from the Original 
Yao Protocol mentioned in Section \ref{prot_prez}. The first remark 
is that the fact that the Garbler sends multiple 
copies of its keys is what allows the success probability to be constant and 
independent from the size of the inputs (see Theorem \ref{correct_att}). Otherwise it would decrease exponentially with the 
number of entries in the garbles table, which might not be too bad if it is a 
small constant (the minimum being $4$ for the naive implementation). On the 
other hand, returning the Garbler's keys to the Adversary is an essential part 
of the attack, as otherwise it would not be able to correct them (the final 
operations described in the full attack are all performed on these 
registers). If they stay in the hands of the Evaluator, it is unclear 
how the Adversary would perform the attack (as the state is then similar to 
the entangled one described in the introduction as something that we seek to avoid). 
Similarly, the fact that we do not use an IND-CPA secure symmetric encryption 
scheme is linked to the fact that it adds an additional register containing 
the randomness used to encrypt (for quantum notions of IND-CPA developed in 
\cite{BZ13} and \cite{MS16}), and which is then entangled with the rest of the 
state (this register can not be given back to the Adversary as it would break 
the security, even in the classical case, by revealing the index of the correctly decrypted 
garbled entry). On the other hand, in \cite{GHS16} they show that the notion 
of quantum IND-CPA they define is impossible for 
\emph{quasi-length-preserving} encryption scheme (which is equivalent to 
\emph{format-preserving} from Definition \ref{format-pres}). Finally, if we 
were to follow the same principle as in the original protocol and decompose 
the binary function into separate gates, then the intermediate keys would 
similarly add another register which is entangled with the rest of the state. 
This is why we require that the garbled table represents the whole function.

%% file: tex_files/A6_insec_proofs.tex
We present here the proof of the Theorems from Section \ref{sec:attack}.

\subsubsection{Proof of Theorem \ref{correct_att}}
\label{subsec:proof_sgc}

\setcounter{tempresult}{\value{theorem}}
\setcounter{theorem}{\value{compteur:correct_att}-1}

\begin{theorem}[State Generation Analysis]
The state contained in the Garbler's attack registers at the end of a 
successful Superposition Generation Procedure (Attack \ref{Yao_att}) is negligibly close to 
$\frac{1}{2}\sum\limits_{x, k^z} (-1)^{k^z}\ket{x^{L'}}\ket{f(x, \hat{y}) \oplus 
k^z}$, where $x^{L'}$ is a logical encoding of $x$ and $x \in \{\widehat{x_0}, 
\widehat{x_1}\}$. Its success probability is lower bounded by $1 - e^{-1}$ for all values of $n_X$ 
and $n_Y$.
\end{theorem}

We prove the two parts of Theorem \ref{correct_att} separately, first analysing the result of a successful execution and later computing the success probability of the procedure.

\begin{proof}[State Generation Correctness (Theorem \ref{correct_att}, part 1)]

The state in the registers of the Evaluator (with input $\hat{y} \in \{0, 1\}^{n_Y}$) before it starts the decryption process is (up to appropriate normalisation):

\[
\ket{\hat{y}} \otimes \ket{k^E_{\hat{y}}} \otimes \frac{1}{\sqrt{2}}\bigl(\ket{k_{\widehat{x_0}}^G} +\ket{k_{\widehat{x_1}}^G}\bigr) \otimes \ket{\widetilde{\mathrm{aux}}} \otimes \smashoperator[r]{\sum\limits_{x, y}}\ket{E_{x, y}^0} - \ket{E_{x, y}^1}
\]

In fact there are $2^{n_X + n_Y}$ registers containing the superposition of keys and the same number containing the superposition of encryptions, but it suffices to consider the result on one such register (the protocol has been specifically tailored so that repetitions can be handled separately, as seen in the next part of the proof). For $x \neq x'$ or $y \neq y'$ (inclusively), let $g_{x, y}^{x', y', k^z} = \widetilde{\Dec}_{k^G_{x}, k^E_{y}}(\widetilde{\mathrm{aux}}, E_{x', y'}^{k^z})$ (this is the decryption of $E_{x', y'}^{k^z}$ using the keys for $x$ and $y$, leading to a wrong decryption as at least one key does not match and generating the garbage value $g_{x, y}^{x', y', c}$). The state after applying the decryption procedure is then (for $x \in \{\widehat{x_0}, \widehat{x_1}\}$):

\[
\ket{C} \otimes \biggl(\sum\limits_{x, k^z} (-1)^{k^z}\ket{\widetilde{d_K}\qty(k^G_x)}\ket{f(x, \hat{y}) \oplus k^z}\ket{0}^{\otimes p} + \smashoperator{\sum\limits_{\substack{k^z, x, x', y' \\ (x, y) \neq (x', \hat{y})}}}(-1)^c\ket{\widetilde{d_K}\qty(k^G_x)}\ket{g_{x, \hat{y}}^{x', y', k^z}}\biggr)
\]

Here the registers containing the Garbler's keys have been rearranged and 
$\ket{C} = \ket{\hat{y}} \otimes 
\ket{\widetilde{d_K}\bigl(k^E_{\hat{y}}\bigr)} \otimes 
\ket{\widetilde{d_A}(\widetilde{\mathrm{aux}})}$ corresponds to the classical 
values unentangled from the rest of the state. With overwhelming probability 
in $\eta$ (based on the analysis from Theorem \ref{prot_corr}), there are no 
values $(r, c, x, x', y')$ such that $g_{x, \hat{y}}^{x', y', c} = r \parallel 
0^p$ and so the states 
\[\sum\limits_{x, k^z} 
(-1)^{k^z}\ket{\widetilde{d_K}(k^G_x)}\ket{f(x, \hat{y}) \oplus 
k^z}\ket{0}^{\otimes p}\] 
and 
\[\smashoperator{\sum\limits_{\substack{k^z, x, 
x', y' \\ (x, y) \neq (x', 
\hat{y})}}}(-1)^{k^z}\ket{\widetilde{d_K}(k^G_x)}\ket{g_{x, \hat{y}}^{x', y', 
k^z}}\]
are orthogonal. If the measurement $\mathcal{M}_C$ succeeds (ie. the 
outcome is $\dyad{0^p}$), the projected state is (also up to appropriate 
normalisation):

\[
\ket{C} \otimes \sum\limits_{x, k^z} (-1)^{k^z}\ket{\widetilde{d_K}(k^G_x)}\ket{f(x, \hat{y}) \oplus k^z}\ket{0}^{\otimes p}
\]

Note that the keys of the Evaluator and the auxiliary values are unentangled from the rest of the state during the whole process thanks to the properties satisfied by the symmetric encryption scheme. The state in the Garbler's registers after receiving the output and its keys is then simply:

\[
\sum\limits_{x, k^z} (-1)^{k^z}\ket{\widetilde{d_K}(k^G_x)}\ket{f(x, \hat{y}) \oplus k^z}
\]

After applying the first step of clean-up procedure at the end (applying the inverse of $d_K$ for each key), the Garbler is left with the state:

\[
\sum\limits_{x, k^z} (-1)^c\ket{k^G_x}\ket{f(x, \hat{y}) \oplus k^z}
\]

To demonstrate the effect of the rest of the clean-up procedure, we will apply it to an example with $k_0 = 01110$ and $k_1 = 11100$ (for an Adversary's input consisting of a single bit). The corresponding (non-normalised) superposition is then $\ket{k_0}\ket{f(0, \hat{y})} + \ket{k_1}\ket{f(1, \hat{y})} = \ket{01110}\ket{f(0, \hat{y})} + \ket{11100}\ket{f(1, \hat{y})}$ (the terms with $k^z = 1$ behave similarly). If the bits of key are the same, we can factor out the corresponding qubits (in this case, the second, third and fifth qubits are unentangled from the rest). This gives the state $\ket{110} \otimes (\ket{01}\ket{f(0, \hat{y})} + \ket{10}\ket{f(1, \hat{y})})$. The unentangled qubits may be discarded and then the qubits $i$ for which $k_0^i \neq k_1^i$ and $k_0^i = 1$ are flipped using $\X$ (meaning the fourth initial qubit in this case, or the second one after discarding the unentangled qubits). The result is $\ket{00}\ket{f(0, \hat{y})} + \ket{11}\ket{f(1, \hat{y})}$. This procedure does not depend on the choice of $\hat{y}$ (and is the same for $k^z = 1$), only on the keys that were generated by the Adversary.

In the general case, the final clean-up transforms each key associated with a bit-value of $0$ into a logical $0$ (ie. $0^{L'_i}$ for a random but known value $L'_i$), and similarly with the corresponding key associated to the bit-value $1$ (changed into $1^{L'_i}$ with the same $L'_i$). The final result is therefore (where $x^{L'}$ is a logical encoding of $x$ where some bits may be repeated a variable but known number of times):

\[
\frac{1}{2}\sum\limits_{x, k^z} (-1)^{k^z}\ket{x^{L'}}\ket{f(x, \hat{y}) \oplus k^z}
\]

This is exactly the state that was expected, therefore concluding the proof.
\qed
\end{proof}

%
%

\begin{proof}[Success Probability of State Generation (Theorem \ref{correct_att}, part 2)]

If a given measurement fails, based on the analysis in the previous proof, the state in the Evaluator's registers corresponding to this decryption is negligibly close to:

\[
\ket{\hat{y}} \otimes \ket{\tilde{d}_K(k^E_{\hat{y}})} \otimes \ket{\tilde{d}_A(\widetilde{\mathrm{aux}})} \otimes \smashoperator{\sum\limits_{\substack{k^z, x, x', y' \\ (x, y) \neq (x', \hat{y})}}}(-1)^{k^z}\ket{\widetilde{d_K}(k^G_x)}\ket{g_{x, \hat{y}}^{x', y', k^z}}
\]

By applying the inverse of the $d_K$ and $d_A$ operations on each of the registers containing its keys and the auxiliary values, the Evaluator recovers the state:

\[
\ket{\hat{y}} \otimes \ket{k^E_{\hat{y}}} \otimes \ket{\widetilde{\mathrm{aux}}} \otimes \smashoperator{\sum\limits_{\substack{k^z, x, x', y' \\ (x, y) \neq (x', \hat{y})}}}(-1)^{k^z}\ket{\widetilde{d_K}(k^G_x)}\ket{g_{x, \hat{y}}^{x', y', k^z}}
\]

Unless it is the last remaining copy of the superposition of Garbler's keys and garbled values (in which case the attack has failed), the Evaluator can simply proceed and repeat the decryption process using its keys and the auxiliary values on the next copy (the failed decryption state is unentangled from the rest and can be ignored in the remaining steps). This essentially means that the Evaluator has $2^{n_X + n_Y}$ independent attempts to obtain measurement result $0^p$.

Since the states that are being considered are normalised and in a uniform superposition, the probability of success of each measurement attempt is simply given by the number of states correctly decrypted out of the total number of states.

There are $2^{n_X + n_Y + 1}$ encrypted values in the garbled table and $2$ key pairs (one key for wires in $Y$ and $2$ keys for wires in $X$). There are therefore $2^{n_X + n_Y + 2}$ decrypted values (taking into account decryptions performed with the incorrect keys and counting duplicates). For each key pair, there are exactly two ciphertexts which will decrypt correctly (one for each value of $k^z$), meaning that $4$ decrypted values out of $2^{n_X + n_Y + 2}$ have their last $p$ bits equal to $0$. The probability of the measurement $\mathcal{M}_C$ succeeding is therefore $\frac{1}{2^{n_X + n_Y}}$. The probability that no measurement succeeds in $2^{n_X + n_Y}$ independent attempts (noted as event $\mathsf{Fail}$) is given by:

\[
\mathbb{P}[\mathsf{Fail}] = \biggl(1 - \frac{1}{2^{n_X + n_Y}}\biggr)^{2^{n_X + n_Y}}
\]

The function $p(x) = (1 - \frac{1}{x})^x$ is strictly increasing and upper-bounded by $e^{-1}$, meaning that the success probability is $\mathbb{P}[\mathsf{Succ}] = 1 - \mathbb{P}[\mathit{Fail}] \geq 1 - e^{-1}$

\qed
\end{proof}

\subsubsection{Proof of Theorem \ref{sup_att_analysis}}
\label{subsec:proof_att_ana}

\setcounter{tempresult}{\value{theorem}}
\setcounter{theorem}{\value{compteur:sup_att_analysis}-1}
\begin{theorem}[Vulnerability to Superposition Attacks of the Modified Yao Protocol]
For any non-trivial two-party function $f : \{0, 1\}^{n_X} \times \{0, 
1\}^{n_Y} \rightarrow \{0, 1\}$, let $(\widehat{x_0}, \widehat{x_1})$ be a
pair of non-trivial values in $\{0, 1\}^{n_X}$ (they do not determine the
value of the output uniquely). For all inputs $\hat{y}$ of honest Evaluator in
Protocol \ref{Yao}, let $\PredE(\hat{y}) = f(\widehat{x_0}, \hat{y}) \oplus
f(\widehat{x_1}, \hat{y})$. Then there exists a real-world $\BQP$ Adversary
$\mathcal{A}$ against Protocol \ref{Yao} implementing $f$ such that for any
$\BQP$ Simulator $\mathcal{S}$, the advantage of the Adversary over the
Simulator in guessing the value of $\PredE(\hat{y})$ is lower-bounded by
$\frac{1}{2}(1 - e^{-1})$.
\end{theorem}
\setcounter{theorem}{\value{tempresult}}

\begin{proof}[Vulnerability to Superposition Attacks of the Modified Yao Protocol (Theorem \ref{sup_att_analysis})]
Let $(\widehat{x_0}, \widehat{x_1})$ be a pair of values in $\{0, 1\}^{n_X}$ such that there exists $(\widehat{y_0}, \widehat{y_1})$ with $f(\widehat{x_0}, \widehat{y_0}) = f(\widehat{x_1}, \widehat{y_0})$ and $f(\widehat{x_0}, \widehat{y_1}) \neq f(\widehat{x_1}, \widehat{y_1})$ (at least one such pair of inputs exists, otherwise the function is trivial). The Environment $\mathcal{Z}$ initialises the input of the Adversary with values a pair of such values $\widehat{x_0}$ and $\widehat{x_1}$. Let $\hat{y} \in \{0, 1\}^{n_Y}$ be the value of the honest player's input chosen (uniformly at random) by the Environment $\mathcal{Z}$. The goal of the attack is to obtain the value of $\PredE(\hat{y}) = f(\widehat{x_0}, \hat{y}) \oplus f(\widehat{x_1}, \hat{y})$.

The Adversary will try to generate the superposition state during the protocol using Attack \ref{Yao_att}, succeeding with probability $p_{\mathit{Gen}}$. If the state has been generated correctly, Adversary will apply the final steps of Deutsch's algorithm and recover the value of the XOR with probability equal to $1$ (see below). If the state generation fails, Adversary resorts to guessing the value of the value of $\PredE(\hat{y})$, winning with a probability $p_{\mathit{Guess}}$. On the other hand, the Simulator is only able toss a coin to guess the value of $\PredE(\hat{y})$ (the only information that it possesses is either $f(\widehat{x_0}, \hat{y})$ or $f(\widehat{x_1}, \hat{y})$, given by the Ideal Functionality), winning with probability $p_{\mathit{Guess}}$.

The overall advantage of the Adversary is therefore $p_{\mathit{Gen}} \cdot (1 - p_{\mathit{Guess}})$ (if the State Generation Procedure does not succeed, the probabilities of winning of the Adversary and the Simulator are the same). It has been shown via Theorem \ref{correct_att} that the probability of generating the state is lower-bounded by $1 - e^{-1}$, the rest of the proof will focus on describing the last steps of the Attack \ref{Yao_full_att} and calculating the other values defined above. 

We first analyse the behaviour of the state during the Adversary's calculation in Attack \ref{Yao_full_att} if there was no $\Abort$. The state in the register of the Adversary registers at the end of a successful Superposition Generation Procedure via Attack \ref{Yao_att} is (the logical encoding $L'$ being known to the Adversary):

\begin{multline}
\frac{1}{2}\Bigl(\ket{\widehat{x_0}^{L'}}\ket{f(\widehat{x_0}, \hat{y})} - \ket{\widehat{x_0}^{L'}}\ket{f(\widehat{x_0}, \hat{y}) \oplus 1} \\
+ \ket{\widehat{x_1}^{L'}}\ket{f(\widehat{x_1}, \hat{y})} - \ket{\widehat{x_1}^{L'}}\ket{f(\widehat{x_1}, \hat{y}) \oplus 1}\Bigr)
\end{multline}

The Adversary applies the clean-up procedure on the registers containing $\widehat{x_i}^{L'}$ and obtains (for a different value $L$ for the logical encoding):

\[
\frac{1}{2}\bigl(\ket{0}^{\otimes L}\ket{f(\widehat{x_0}, \hat{y})} - \ket{0}^{\otimes L}\ket{f(\widehat{x_0}, \hat{y}) \oplus 1} + \ket{1}^{\otimes L}\ket{f(\widehat{x_1}, \hat{y})} - \ket{1}^{\otimes L}\ket{f(\widehat{x_1}, \hat{y}) \oplus 1}\bigr)
\]

This is exactly the state of Deutsch's algorithm after applying the (standard) oracle unitary implementing $U_{f_{\hat{y}}^{\widehat{x_0}, \widehat{x_1}}}$, where $f_{\hat{y}}^{\widehat{x_0}, \widehat{x_1}}(b) = f(\widehat{x_b}, \hat{y})$ (by standard we mean of the form $U_f\ket{x}\ket{b} = \ket{x}\ket{b \oplus f(x)}$, in comparison to the Minimal Oracle Representation). The rest of the attack and analysis follows the same pattern as Deutsch's algorithm.

For simplicity's sake, let $b_i := f(\widehat{x_i}, \hat{y})$, then the state is:

\[
\frac{1}{\sqrt{2}}(-1)^{b_0}\bigl(\ket{0}^{\otimes L} + (-1)^{b_0 \oplus b_1}\ket{1}^{\otimes L}\bigr)\otimes\ket{-}
\]

The Adversary then applies the logical Hadamard gate, the resulting state is (up to a global phase):

\[
\ket{b_0 \oplus b_1}^{\otimes L}\otimes\ket{-}
\]

The Adversary can measure the first qubit in the computational basis and distinguish perfectly both situations, therefore obtaining $f(\widehat{x_0}, \hat{y}) \oplus f(\widehat{x_1}, \hat{y}) = b_0 \oplus b_1$.

On the other hand, in the ideal scenario, to compute the probability of guessing the correct answer $p_{\mathit{Guess}}$, we consider the mixed strategies in a two-player game between the Environment $\mathcal{Z}$ and the Simulator where both players choose a bit simultaneously, the Simulator wins if they are the same and the Environment wins if they are different (this represents the most adversarial Environment for the Simulator). The Environment chooses bit-value $0$ with probability $p$, while the Simulator chooses the bit-value $0$ with probability $q$. The probability of winning for the Simulator is then $p_{\mathit{Guess}} = pq + (1-p)(1-q) = 1 - q - p(1 - 2q)$. We see that if $q \neq \frac{1}{2}$ there is a pure strategy for the Environment such that $p_{\mathit{Guess}} < \frac{1}{2}$ (if the Simulator chooses its bit in a way that is biased towards one bit-value, the Environment always chooses the other), while if $q = \frac{1}{2}$ then $p_{\mathit{Guess}} = \frac{1}{2}$. The same analysis can be applies to the Environment and therefore $p = \frac{1}{2}$ as well.

In the end, we have $p_{\mathit{Guess}} = \frac{1}{2}$ and therefore the advantage of the Adversary is $\mathsf{Adv} = p_{\mathit{Gen}} (1 - p_{\mathit{Guess}}) \geq \frac{1}{2}(1 - e^{-1})$, which concludes the proof.

\qed
\end{proof}

As a remark, the inverse of the circuit preparing the GHZ state can be applied as the final step of the attack instead of the logical Hadamard $\Ha_L$, yielding the same result (the state is then $\ket{b_0 \oplus b_1} \otimes \ket{0}^{\otimes L-1}$ and measuring the first qubit in the computational basis still gives the correct value). The presentation above was chosen to closely reflect the description of Deutsch's algorithm.

%% file: tex_files/A7_sup_sec_yao.tex
We give here a sketch of the formal Superposition-Secure Yao Protocol \ref{Yao_sup}, along with a proof of its security against an adversarial Garbler with superposition access.

\begin{algorithm}[H]
\caption{Superposition-Secure Yao Protocol (Sketch).}
\label{Yao_sup}
\begin{algorithmic}[0]
\STATE \textbf{Input:} The Garbler and Evaluator have inputs $x \in \{0, 1\}^{n_X}$ and $y \in \{0, 1\}^{n_Y}$ respectively.
\STATE \textbf{Output:} The Garbler has no output, the Evaluator has one bit of output.
\STATE \textbf{Public Information:} The function $f$ to be evaluated, the encryption scheme $(\Enc, \Dec)$ and the size of the padding $p$.
\STATE \textbf{The Protocol:}
\begin{enumerate}
\item The Garbler creates the keys and garbled table as in the original version.
\item The Garbler and the Evaluator participate in the OT ideal executions, at the end of which the Evaluator receives its evaluation keys for its input of choice.
\item The Garbler sends the evaluation keys for its inputs and stops (it has no further actions and no output)
\item The Evaluator decrypts each entry in the garbled table sequentially, stopping if the padding is $0^p$ at the end of a decrypted value. The rest of the decrypted value is then set as its output.
\item Otherwise (if it has not stopped at the previous step, meaning that none of the values were decrypted correctly), it sets as its output $\Abort$. The fact that it has aborted or not is not communicated to the Garbler
\end{enumerate}
\end{algorithmic}
\end{algorithm}

\setcounter{tempresult}{\value{theorem}}
\setcounter{theorem}{\value{compteur:Yao-sup-sec}-1}
\begin{theorem}[$\BQP$-Security of Superposition-Resistant Yao Protocol]
The Superposition-Resistant Yao Protocol \ref{Yao_sup} is perfectly-secure against a 
$\BQP$ adversarial Garbler according to Definition \ref{sec-def} in an OT-hybrid execution.
\end{theorem}
\setcounter{theorem}{\value{tempresult}}

\begin{proof}[$\BQP$-Security of Superposition-Resistant Yao Protocol (Theorem \ref{Yao_sup_sec})]

The Garbler cannot break the security of the OT ideal execution, which 
furthermore is classical. The rest of the protocol can be summarised by the 
Garbler sending one quantum state and then the Evaluator performing a local 
operation on it and stopping. This is exactly the same scenario as in the 
One-Time Pad protocol and the same analysis applies in this case: the Garbler 
recovering any information from this operation would violate the no-signalling 
principle. The resulting security bound is therefore $0$.

\qed
\end{proof}

The proof above does not translate into a proof for an actual instance of the protocol since security in our model does not hold under sequential composability, but it gives a hint as to which steps are crucial for securing it. Another path for obtaining security could be to replace the encryption scheme with one for which there is no efficient Minimal Oracle Representation. We leave this case as an open question.

%% file: tex_files/06_attack_OT.tex
The attack described in Section \ref{sec:attack} will now be applied to a simple function, namely the 1-out-of-2 bit-OT, in order to demonstrate a potential improvement. In this case, the Garbler has a bit $b$ as input, the Evaluator has two bits $(x_0, x_1)$ and the output for the Garbler is $x_b$. This can be represented by the function $\mathcal{OT}(b, x_0, x_1) = bx_1 \oplus (1 \oplus b)x_0$. This can be factored as $\mathcal{OT}(b, x_0, x_1) = b(x_0 \oplus x_1) \oplus x_0$. By changing variables and defining $X := x_0 \oplus x_1$, it can be rewritten further into $\mathcal{OT}(b, x_0, X) = bX \oplus x_0$.

Based on this simplified formula, instead of computing the garbled table for the full function, the Garbler will only garble the AND gate between $b$ and $X$. In order to compute the XOR gate at the end, the Free-XOR technique will be used. Recall first that the key-space is fixed to $\mathfrak{K} = \{0, 1\}^{n_K}$. Instead of choosing both keys for each wire uniformly at random, this technique works by choosing uniformly at random a value $K \in \{0, 1\}^{n_K}$ and setting $k_1^w := k_0^w \oplus K$ for all wires $w$ which are linked to the XOR gate (either as input or output wires). The value $k_0^w$ is sampled uniformly at random for the input wires. For the output wire, if $a$ and $b$ are the labels of the input wires, the value is set to $k_0^w = k_0^a \oplus k_0^b$. In this way, instead of going through the process of encrypting and then decrypting a garbled table, given a key for each input of a XOR gate, the Evaluator can directly compute the output key in one string-XOR operation (as an example, if the keys recovered as inputs for the input wires are $k_0^a$ and $k_1^b$, then the output key is computed as $k_0^a \oplus k_1^b = k_0^a \oplus k_0^b \oplus K = k_0^w \oplus K = k_1^w$, which is the correct output key value for inputs $a = 0$ and $b = 1$). The security of Yao's protocol using the Free-XOR technique derives from the fact that only one value for the keys is known to the evaluator at any time, so the value $K$ is completely hidden (if the encryption scheme is secure). This has been first formalised in \cite{KS08}.

After having decrypted the garbled table for the AND gate, the Evaluator simply performs the XOR gate using the Free-XOR technique. Without loss of generality the XOR of the keys is performed into the register containing the key corresponding to the output of the AND gate. In the quantum case, this is done using a $\CNOT$ gate, where the control qubit is the register containing the keys for $x_0$ and the controlled qubit is the register containing the output of the decryption of the garbled AND gate (the key for $x_0$ is not in superposition as it belongs to the Evaluator and so the register containing it remains unentangled from the rest on the state).

The initial input to the garbled table is $3$ bits long in the decomposed protocol, while the input to the AND gate is only $2$ bits long, lowering the number of pre-computations to generate the garbled table and improving slightly the attack's success probability (it is a decreasing function of the number of possible inputs).

The probability of successfully generating the attack superposition $\frac{1}{2}\bigl(\ket{0}^{\otimes L}\ket{x_0} - \ket{0}^{\otimes L}\ket{x_0 \oplus 1} + \ket{1}^{\otimes L}\ket{x_1} - \ket{1}^{\otimes L}\ket{x_1 \oplus 1}\bigr)$ by using this new technique is $1 - \Bigl(\frac{3}{4}\Bigr)^4 = \frac{175}{256}$ (by not using the approximation at the end of the proof of part 2 of Theorem \ref{correct_att} for success probability). As described in Theorem \ref{sup_att_analysis}, such a superposition can been used to extract the XOR of the two values, an attack which is impossible in the classical setting or even in the quantum setting without superposition access. The advantage of the Adversary in finding the XOR (over a Simulator which guesses the value) by using this attack is $\frac{175}{512}$. This is far from negligible and therefore the security property of the OT is broken.

Of course this is a toy example as it uses two \emph{string}-OTs to generate one \emph{bit}-OT. But the bit-OT that has been generated has a reversed Sender and Receiver compared to the string-OTs that were used. In the classical case, it can be noted that similar constructions have been proposed previously to create an OT which was simulatable for one party based on an OT that is simulatable for the other (and this construction is close to round-optimal).